\documentclass{vldb}
\usepackage{graphicx}
\usepackage{balance,cite}  
\usepackage{url,xspace}
\usepackage{enumitem}

\def\full{1}        
\def\shownotes{1}   

\newtheorem{theorem}{Theorem}
\newtheorem{definition}{Definition}

\newtheorem{claim}[theorem]{Claim}

\ifnum\shownotes=0
\newcommand{\authnote}[2]{{ $\ll$\textsf{\footnotesize #1 notes: #2}$\gg$}}
\else
\newcommand{\authnote}[2]{}
\fi
\newcommand{\Snote}[1]{{\authnote{Sharon}{#1}}}

\newif\ifdsexp
\dsexpfalse

\providecommand{\ie}{\emph{i.e.,} }
\providecommand{\eg}{\emph{e.g.,} }

\providecommand{\etal}{\emph{et al.}\xspace}   
\providecommand{\etc}{\emph{etc.}}      

\providecommand{\myparab}[1]{\smallskip\noindent\textbf{#1} }

\title{Calibrating Data to Sensitivity in Private Data Analysis}
\subtitle{A Platform for  Differentially-Private Analysis of Weighted Datasets
\ifnum\full=1
\begin{large} . \\ Full version. Last updated \today. \end{large}
\fi
}

\numberofauthors{3}
\author{
\alignauthor
Davide Proserpio\\
       \affaddr{Boston University}\\
       \email{dproserp@bu.edu}
\alignauthor
Sharon Goldberg\\
       \affaddr{Boston University}\\
       \email{goldbe@cs.bu.edu}
\alignauthor
Frank McSherry\\
       \affaddr{Microsoft Research}\\
       \email{mcsherry@microsoft.com}
}
\date{\today}

\begin{document}
\maketitle
\begin{abstract}

We present an approach to differentially private computation in which one does not scale up the magnitude of noise for challenging queries, but rather scales {\em down} the contributions of challenging records. While scaling down all records uniformly is equivalent to scaling up the noise magnitude, we show that scaling records {\em non-uniformly} can result in substantially higher accuracy by bypassing the worst-case requirements of differential privacy for the noise magnitudes.

This paper details the data analysis platform \textbf{wPINQ}, which generalizes the Privacy Integrated Query (PINQ) to weighted datasets.
Using a few simple operators (including a non-uniformly scaling Join operator) wPINQ can reproduce (and improve) several recent results on graph analysis and introduce new generalizations (\eg counting triangles with given degrees). We also show how to integrate probabilistic inference techniques to synthesize datasets respecting more complicated (and less easily interpreted) measurements.

\end{abstract}



\section{Introduction}\label{sec:intro}

Differential Privacy (DP) has emerged as a standard for privacy-preserving data analysis.  A number of platforms propose to lower the barrier to entry for analysts interested in differential privacy by presenting languages that guarantee that all written statements satisfy differential privacy~\cite{PINQ,airavat,fire,DFuzz,gupt}.  However, these platforms have limited applicability to a broad set of analyses, in particular to the analysis of graphs, because they rely on DP's {worst-case sensitivity} bounds over multisets.

In this paper we present a platform for differentially private data analysis, wPINQ (for ``weighted" PINQ), which uses weighted datasets to bypass many  difficulties encountered when working with worst-case sensitivity.
wPINQ follows the language-based approach of PINQ~\cite{PINQ}, offering a SQL-like declarative analysis language, but extends it to a broader class of datasets with more flexible operators, making it capable of graph analyses that PINQ, and other differential privacy platforms, are unable to perform.
wPINQ also exploits a connection between DP and incremental computation to provide a random-walk-based probabilistic inference engine, capable of fitting synthetic datasets to arbitrary wPINQ measurements.
Compared to PINQ and other platforms, wPINQ is able to express and automatically confirm the privacy guarantees of a richer class of analyses, notably graph analyses, and automatically invoke inference techniques to synthesize representative datasets.

In an earlier workshop paper~\cite{wosn}, we sketched a proposed workflow for DP graph analysis and presented preliminary results showing how these techniques could be used to synthesize graphs that respect degree and joint-degree distributions.  This paper is a full treatment of our data analysis platform, including the wPINQ programming language and its formal properties, as well as its incremental query processing engine for probabilistic inference. We discuss why wPINQ is well-suited to solve problems in social graph analysis, and present new and more sophisticated use cases related to computing the distribution of triangles and motifs in social graphs, results of independent interest.

\subsection{Reducing sensitivity with weighted data}

In the worst-case sensitivity framework, DP techniques protect privacy by giving noisy answers to queries against a dataset; the amplitude of the noise is determined by the sensitivity of the query. In a sensitive query, a single change in just one record of the input dataset can cause a large change in the query's answer, so substantial noise must be added to the answer to mask the presence or absence of that record. Because worst-case sensitivity is computed as the maximum change in the query's answer over \emph{all possible input datasets}, this can result in significant noise, even if the input dataset is {not} worst-case.
This problem is particularly pronounced in analysis of social graphs, where the presence or absence of a single edge can drastically change the result of a query~\cite{hayJoins,smooth}. Consider the problem of counting triangles in the two graphs of Figure~\ref{fig:triangles-intro}. The left graph
has no triangles, but the addition of just one edge $(1,2)$ immediately creates $|V|-2$ triangles; to achieve differential privacy, the magnitude of the noise added to the query output is proportional to $|V|-2$. This same amount of noise must be added to the output of the query on the right graph, even though it is not a ``worst case'' graph like the first graph.

We will bypass the need to add noise proportional to worst-case sensitivity by working with \emph{weighted datasets}.  Weighted datasets are a generalization of multisets, where records can appear an integer number of times, to sets in which records may appear with a real-valued multiplicity. Weighted datasets allows us to smoothly suppress individual `troublesome' records (\ie records that necessitate extra noise to preserve privacy) by scaling down their influence in the output. We scale down the weight of these individual troublesome output records by the minimal amount they \emph{do} require, rather than scale up the amplitude of the noise applied to all records by the maximal amount they \emph{may} require. This \emph{data-dependent rescaling} introduces inaccuracy only when and where the data call for it, bypassing many worst-case sensitivity bounds, especially for graphs.

\begin{figure}[t]
\begin{center}
\includegraphics[scale=0.3]{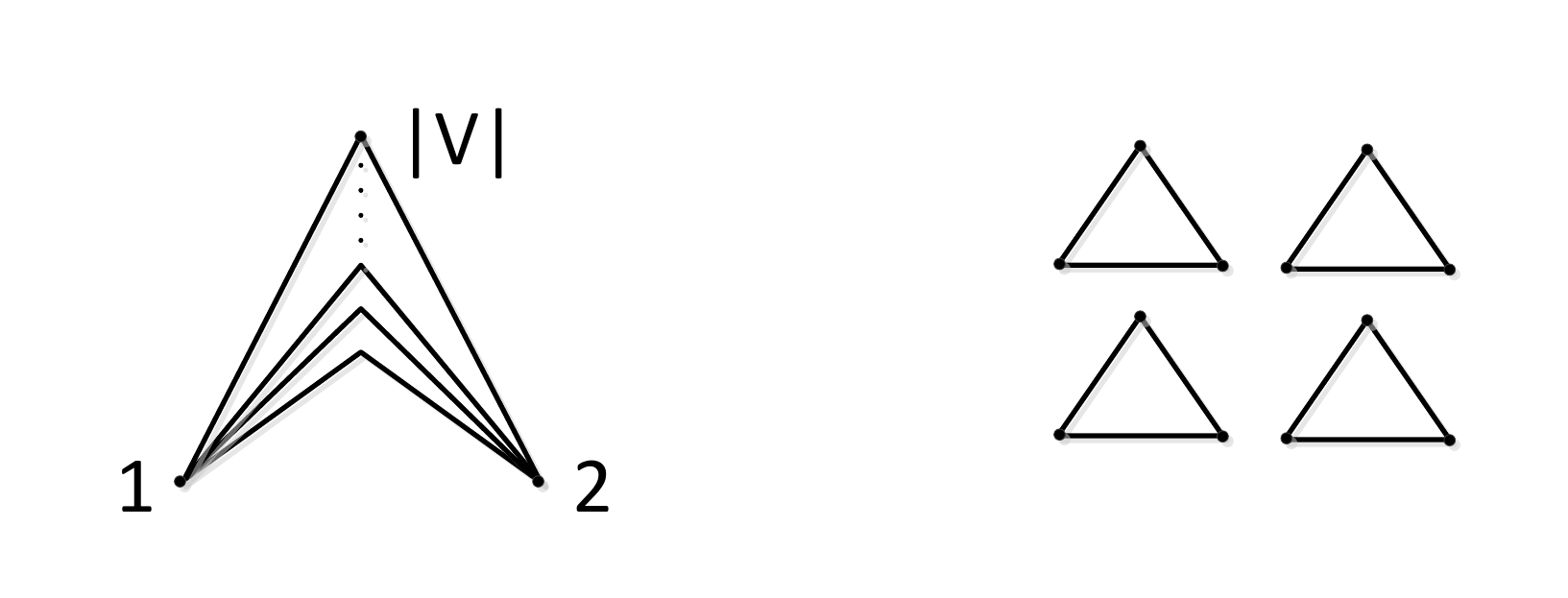}
\vspace{-5mm}
\caption{(left) Worst- and (right) best-case graphs for the problem of privately counting triangles.}\label{fig:triangles-intro}
\vspace{-8mm}
\end{center}
\end{figure}

Returning to the example in Figure~\ref{fig:triangles-intro}, if the weight of each output triangle $(a,b,c)$ is set to be $1/\max \{ d_a, d_b, d_c \}$ where $d_a$ is the degree of node $a$, the presence or absence of any single input edge can alter only a constant total amount of weight across all output triangles. (This is because any edge $(a,b)$ can create at most $\max\{d_a,d_b\}$ triangles.) 
 %
 %
Adding the weights of these triangles (the weighted analog of ``counting") with constant-magnitude noise provides differential privacy, and can result in a more accurate measurement than counting the triangles with noise proportional to $|V|$ (which is equivalent to weighting all triangles with weight $1/|V|$).
%
%
While this approach provides no improvement for the left graph, it can significantly improve accuracy for the right graph: since this graph has constant degree, triangles are measured with only constant noise.

It is worth noting how this approach differs from smooth sensitivity~\cite{smooth}, which adds noise based on instance-dependent sensitivity; if the left and right graphs of Figure~\ref{fig:triangles-intro} were unioned into one graph, smooth sensitivity would still insist on a large amount of noise, whereas weighted datasets would allow the left half to be suppressed while the right half is not. This approach also differs from general linear queries allowing non-uniform weights~\cite{matrixAlg2}; these approaches apply only to linear queries (triangles is not) and require the non-uniformity to be explicitly specified in the query, rather than determined in a data-dependent manner. Weighted datasets are likely complementary to both of these other approaches.


\subsection{Our platform: \lowercase{w}PINQ + MCMC}

In Section~\ref{sec:wPINQ}, we describe the design of wPINQ, a declarative programming language over weighted datasets, which generalizes PINQ~\cite{PINQ}. While PINQ provided multiset transformations such as \texttt{Select}, \texttt{Where}, and \texttt{GroupBy}, the limitations of multisets meant that it lacked an effective implementation of the \texttt{Join} transformation, among others.  wPINQ extends these to the case of weighted datasets, and uses data-dependent rescaling of record weights to enable a new, useful, \texttt{Join} operator, as well as several other useful transformations (\eg \texttt{Concat}, \texttt{SelectMany}, \texttt{Union}, \etc).


In Section~\ref{sec:graphs} we show how wPINQ operators can implement a number of new graph analyses, focusing on subgraph-counting queries that informed the graph generator in~\cite{sigcomm06}.  We use a few lines of wPINQ to produce new algorithms to count triangles incident on vertices of degrees $(d_1, d_2, d_3)$ each with noise proportional to $O(d_1^2 + d_2^2 + d_3^2)$, as well as length-4 cycles incident on vertices of degrees $(d_1, d_2, d_3, d_4)$.




In Section~\ref{sec:pi} we show how wPINQ measurements can be analyzed with Markov Chain Monte Carlo (MCMC) techniques, to sample a synthetic dataset from the posterior distributions over datasets given wPINQ's noisy observations. This post-processing serves three purposes:
\begin{enumerate}[itemsep=0pt]
\item It can improve the accuracy of individual answers to wPINQ queries by removing obvious inconsistencies introduced by noise (\eg when the noisy count of triangles is a negative number, or not a multiple of six).
\item It can improve the accuracy of multiple measurements by combining the constraints they impose on each other (\eg the degree distribution and joint-degree distribution constrain each other; fitting a synthetic graph to both measurements produces more accurate results).
\item It can provide useful estimates for quantities that have not been directly queried in wPINQ, by evaluating them on the samples from the posterior distribution (\eg the joint-degree distribution constrains a graph's assortativity, \ie the extent to which nodes connect to other nodes with similar degrees, and the assortativity on the sampled graphs should be relatively accurate).
\end{enumerate}
We detail the MCMC process, and the design and implementation of our efficient incremental re-execution platform that speeds up the iterative query re-evaluation at the heart of the MCMC process.


Finally, as a case study of the utility of our platform, Section~\ref{sec:synth} discusses the application of our platform to the problem counting triangles in a graph and evaluates its performance on several datasets.  


\section{Weighted Datasets and \lowercase{w}PINQ}\label{sec:wPINQ}

In order to design differentially-private algorithms that surmount worst-case sensitivity bounds by scaling down the influence of troublesome records, it will be convenient for us to work with \emph{weighted datasets}. Section~\ref{sec:background} discusses differential privacy (DP) for weighted datasets. The remainder of this section presents our design of \textbf{Weighted PINQ (wPINQ)}, a declarative programming language for weighted datasets that guarantees DP for all queries written in the language.
The structure of wPINQ is very similar to its predecessor PINQ~\cite{PINQ}; both languages apply a sequence of \emph{stable transformations} to a dataset (Section~\ref{sec:trans}), and then release results after a \emph{differentially-private aggregation} (Section~\ref{sec:agg}) is performed and the appropriate privacy costs are accumulated.
Our main contribution in the design of wPINQ are new stable transformations operators (Figure~\ref{fig:transformations}) that leverage the flexibility of weighted datasets to rescale individual record weights in a data-dependent manner.  We discuss these transformations in Sections~\ref{sec:select}--\ref{sec:shave}.




\subsection{Weighted Datasets \& Differential Privacy}\label{sec:background}

We can think of a traditional dataset (\ie a multiset) as function $A : D \rightarrow \mathbb{N}$ where $A(x)$ is non-negative integer representing the number of times record $x$ appears in the dataset $A$.  A weighted dataset extends the range of $A$ to the real numbers, and corresponds to a function $A : D \rightarrow \mathbb{R}$ where $A(x)$ is the real-valued weight of record $x$.
%
%
We define the difference between weighted datasets $A$ and $B$ as the sum of their element-wise differences:
$$ \| A - B \| = \sum_x |A(x) - B(x)| \; .$$
We write $\| A \| = \sum_x |A(x)|$ for the size of a dataset.

In subsequent examples, we write datasets as sets of weighted records, where each is a pair of record and weight, omitting records with zero weight. To avoid confusion, we use real numbers with decimal points to represent weights. We use these two datasets in all our examples:
\begin{eqnarray*}
A & = & \{ (``1", 0.75), (``2", 2.0), (``3", 1.0) \}  \\
B & = & \{ (``1", 3.0), (``4", 2.0) \} \; .
\end{eqnarray*}
Here we have $A(``2") = 2.0$ and $B(``0") = 0.0$.

Differential privacy (DP)~\cite{DMNS06} generalizes to weighted datasets:
\begin{definition}
A randomized computation $M$ provides $\epsilon$-{\em differential privacy} if for any weighted datasets $A$ and $B$, and any set of possible outputs $S \subseteq \textrm{Range}(M)$,
\begin{eqnarray*}
\Pr_M[M(A) \in S] & \le & \Pr_M[M(B) \in S] \times \exp(\epsilon \times \|A - B\|)\; .
\end{eqnarray*}
\end{definition}

\noindent
This definition is equivalent to the standard definition of differential privacy on datasets with non-negative integer weights, for which $\|A - B\|$ is equal to the symmetric difference between multisets $A$ and $B$. It imposes additional constraints on datasets with non-integer weights. As in the standard definition, $\epsilon$ measures the \emph{privacy cost} of the computation $M$ (a smaller $\epsilon$ implies better privacy).

This definition satisfies sequential composition: A sequence of computations $M_i$ each providing $\epsilon_i$-DP is itself $\sum_i \epsilon_i$-DP. (This follows by rewriting the proofs of Theorem 3 in \cite{PINQ} to use $\| \cdot \|$ rather than symmetric difference.)  Like PINQ, wPINQ uses this property to track the cumulative privacy cost of a sequence of queries and ensure that they remain below a privacy budget before performing a measurement.

\myparab{Privacy guarantees for graphs.} In all wPINQ algorithms for graphs presented here, the input sensitive dataset is \texttt{edges} is a collection  of edges $(a,b)$ each with weight $1.0$, \ie it is equivalent to a traditional dataset. wPINQ will treat \texttt{edges} as a weighted dataset, and perform $\epsilon$-DP computations on it (which may involve manipulation of records and their weights).  This approach provides the standard notion of ``{edge} differential privacy'' where DP masks the presence or absence of individual edges (as in  \eg~\cite{hay,sala,KRSYgraphs,smooth}). Importantly, even though we provide a standard differential privacy guarantee, our use of weighted datasets allow us to exploit a richer set of transformations and DP algorithms.


Edge differential privacy does not directly provide privacy guarantees for vertices, for example ``{vertex} differential privacy"~\cite{nodeprivacy,chopping,recursive}, in which the presence or absence of entire vertices are masked. While non-uniformly weighting edges in the input may be a good first step towards vertex differential privacy, determining an appropriate weighting for the \texttt{edges} input is non-trivial and we will not investigate it here.


\subsection{Differentially private aggregation}\label{sec:agg}

One of the most common DP mechanisms is the ``noisy histogram'', where disjoint subsets of the dataset are counted (forming a histogram) and independent values drawn from a Laplace distribution (``noise'') are added to each count~\cite{DMNS06}. wPINQ supports this aggregation with the \texttt{NoisyCount}$(A, \epsilon)$ operator, which adds random noise from the $\textrm{Laplace}(1/\epsilon)$ distribution (of mean zero and variance $2/\epsilon^2$) to the weight of {each record} $x$ in the domain of $A$:
\begin{eqnarray*}
\texttt{NoisyCount}(A, \epsilon)(x) & = & A(x) + \textrm{Laplace}(1/\epsilon) \; .
\end{eqnarray*}
\texttt{NoisyCount}$(A,\epsilon)$ provides $\epsilon$-differential privacy, where the proof follows from the proof of \cite{DMNS06}, substituting $\| \cdot \|$ for symmetric difference. Importantly, we do \emph{not} scale up the magnitude of the noise as a function of query sensitivity, as we will instead scale down the weights contributed by records to achieve the same goal.

To preserve differential privacy, \texttt{NoisyCount} must return a noisy value for every record $x$ in the domain of $A$, even if $x$ is not present in the dataset, \ie $A(x)=0$.  With weighted datasets, the domain of $A$ can be arbitrary large. Thus, wPINQ implements \texttt{NoisyCount} with a dictionary mapping only those records with non-zero weight to a noisy count. If \texttt{NoisyCount} is asked for a record $x$ with $A(x)=0$, it returns fresh independent Laplace noise, which is then recorded and reproduced for later queries for the same record $x$.

\myparab{Example.} If we apply \texttt{NoisyCount} to the sample dataset $A$ with noise parameter $\epsilon$, the results for ``0'', ``1'', and ``2''  would be distributed as
\begin{eqnarray*}
\texttt{NoisyCount}(A, \epsilon)(``0") &\sim& 0.00 + \textrm{Laplace}(1/\epsilon) \; ,\\
\texttt{NoisyCount}(A, \epsilon)(``1") &\sim& 0.75 + \textrm{Laplace}(1/\epsilon) \; ,\\
\texttt{NoisyCount}(A, \epsilon)(``2") &\sim& 2.00 + \textrm{Laplace}(1/\epsilon) \; .
\end{eqnarray*}
The result for ``0'' would only be determined (and then recorded) when the value of ``0'' is requested by the user.

\ifnum\full=1
\medskip
Although this paper only requires \texttt{NoisyCount}, several other aggregations generalize easily to weighted datasets, including noisy versions of sum and average. The exponential mechanism~\cite{MT} also generalizes, to scoring functions of that are 1-Lipschitz with respect to an input weighted dataset.
\fi


\subsection{Stable transformations}\label{sec:trans}
\begin{figure}
\begin{eqnarray*}
\texttt{Select} &:& \textrm{per-record transformation}\\
\texttt{Where} &:& \textrm{per-record filtering}\\
\texttt{SelectMany} &:& \textrm{per-record one-to-many transformation} \\
\texttt{GroupBy} &:& \textrm{groups inputs by key} \\
\texttt{Shave} &:& \textrm{breaks one weighted record into several} \\
\texttt{Join} &:& \textrm{matches pairs of inputs by key} \\
\texttt{Union} &:& \textrm{per-record maximum of weights} \\
\texttt{Intersect} &:& \textrm{per-record minimum of weights}
\end{eqnarray*}
\vspace{-8mm}
\caption{Several stable transformations in wPINQ. \label{fig:transformations}}
\vspace{-8mm}
\end{figure}

wPINQ rarely uses the \texttt{NoisyCount} directly on a input weighted dataset, but rather on the output of a \emph{stable transformation} of one weighted dataset to another, defined as:
\begin{definition}\label{def:stable}
A transformation $T : {\mathbb R}^D \rightarrow {\mathbb R}^R$ is \emph{stable} if for any two datasets $A$ and $A'$
\begin{eqnarray*}\nonumber
\|T(A) - T(A')\| & \le & \|A - A'\| \; .
\end{eqnarray*}
A binary transformation $T : ({\mathbb R}^{D_1} \times {\mathbb R}^{D_2}) \rightarrow {\mathbb R}^R$ is \emph{stable} if for any datasets $A, A'$ and $B, B'$
\begin{eqnarray*}\nonumber
\|T(A, B) - T(A', B')\| & \le & \|A - A'\| + \|B - B'\| \;
\end{eqnarray*}
\end{definition}

\noindent
(This definition generalizes Definition 2 in~\cite{PINQ}, again with $\| \cdot \|$ in place of symmetric difference.)  The composition $T_1(T_2(\cdot))$ of stable transformations $T_1,T_2$ is also stable.  Stable transformation are useful because they can composed with DP aggregations without compromising privacy, as shown by the following theorem (generalized from \cite{PINQ}):

\begin{theorem}\label{thm:compose}
If $T$ is stable \emph{unary} transformation and $M$ is an $\epsilon$-differentially private aggregation, then $M(T(\cdot))$ is also $\epsilon$-differentially private.
\end{theorem}
Importantly, transformations themselves do not provide differential privacy; rather, the output of a sequence of transformations is only released by wPINQ after a differentially-private aggregation (\texttt{NoisyCount}), and the appropriate privacy cost is debited from the dataset's privacy budget.

Stability for binary transformations (\eg \texttt{Join}) is more subtle, in that a differentially private aggregation of the output reveals information about both inputs. If a dataset $A$ is used multiple times in a query (\eg as both inputs to a self-join), the aggregation reveals information about the dataset multiple times. Specifically, if dataset $A$ is used $k$ times in a query with an $\epsilon$-differentially-private aggregation, the result is $k\epsilon$-differentially private for $A$. The number of times a dataset is used in a query can be seen statically from the query plan, and wPINQ can use similar techniques as in PINQ to accumulate the appropriate multiple of $\epsilon$ for each input dataset.



The rest of this section presents wPINQ's stable transformations, and discuss how they rescale record weights to achieve stability. We start with several transformations whose behavior is similar to their implementations in PINQ (and indeed, LINQ): \texttt{Select}, \texttt{Where}, \texttt{GroupBy}, \texttt{Union}, \texttt{Intersect}, \texttt{Concat}, and \texttt{Except}. We then discuss \texttt{Join}, whose implementation as a stable transformation 
is a significant departure from the standard relational operator. We also discuss \texttt{Shave}, a new operator that decomposes one record with large weight into many records with smaller weights.

\subsection{Select, {Where}, and SelectMany}\label{sec:select}

We start with two of the most fundamental database operators: \texttt{Select} and \texttt{Where}. \texttt{Select} applies a function $f : D \rightarrow R$ to each input record:
\begin{eqnarray*}
\texttt{Select}(A, f)(x) & = & \sum_{y : f(y) = x} A(y) \; .
\end{eqnarray*}
Importantly, this  produces output records weighted by the \emph{accumulated weight} of all input records that mapped to them. 
\texttt{Where} applies a predicate $p : D \rightarrow \{ 0, 1\}$ to each record and yields those records satisfying the predicate.
\begin{eqnarray*}
\texttt{Where}(A,p)(x) & = & p(x) \times A(x) \; .
\end{eqnarray*}
One can verify that both \texttt{Select} and \texttt{Where} are stable.

\myparab{Example. } Applying \texttt{Where} with predicate $x^2 < 5$ to our sample dataset $A$ in Section~\ref{sec:background} gives $\{ (``1", 0.75), (``2", 2.0) \}$. Applying the \texttt{Select} transformation with $f(x)=x \mod 2$ to $A$, we obtain the dataset $\{ (``0", 2.0), (``1", 1.75) \}$; this follows because the ``1'' and ``3'' records in $A$ are reduced to the same output record (``1'') and so their weights accumulate.

\smallskip
\ifnum\full=0
We also mention the \texttt{SelectMany} operator, that generalizes \texttt{Select} and \texttt{Where}.  (Its full description is in our technical report.)  \texttt{SelectMany} is a unary operator, adapted from LINQ, which allows one-to-many record transformation: \texttt{SelectMany} maps each record to a list of elements, and then outputs the (flattened) collection of all elements from all the produced lists.  To provide a stable \texttt{SelectMany} operator, wPINQ scales down the weight of output records by the number of records produced by the same input record; \eg an input mapped to $n$ items would be transformed into $n$ records with weights scaled down by $n$. Section~\ref{sec:shave} uses \texttt{SelectMany} to transform \texttt{edges} into a dataset of nodes.

\else
Both \texttt{Select} and \texttt{Where} are special cases of the \texttt{SelectMany} operator. The standard \texttt{SelectMany} operator in LINQ maps each record to a list, and outputs the flattened collection of all records from all produced lists. This transformation is not stable; the presence or absence of a single input record results in the presence or absence of as many records as the input record would produce. To provide a stable \texttt{SelectMany} operator, wPINQ scales down the weight of output records by the number of records produced by the same input record; an input mapped to $n$ items would be transformed into $n$ records with weights scaled down by $n$.

Specifically, \texttt{SelectMany} takes a function $f : D \rightarrow {\mathbb R}^R$ mapping each record $x$ to a weighted dataset, which is slightly more general than a list. For stability, we scale each weighted dataset to have at most unit weight and then scale it by the weight of the input record, $A(x)$:
 $$
\texttt{SelectMany}(A, f) = \sum_{x \in D} \left(A(x) \times \frac{f(x)}{\max(1, \|f(x)\|)}\right) \; .
$$

\noindent
The important feature of \texttt{SelectMany} is that different input records may produce different numbers of output records, where the scaling depends only on the number of records produced, rather than a worst-case bound on the number. This flexibility is useful in many cases. One example is in frequent itemset mining: a basket of goods is transformed by \texttt{SelectMany} into as many subsets of each size $k$ as appropriate, where the number of subsets may vary based on the number of goods in the basket.

\Snote{To Frank. I cut this sentence, I don't think it is adding much: ``While this number could be arbitrarily large, it is often not.'' }

\myparab{Example. } If we apply \texttt{SelectMany} to our sample dataset $A$, with function $f(x)=\{1, 2, \ldots, x\}$ where each element in $f(x)$ has weight $1.0$, we get
$$\{ (``1", 0.75 + 1.0 + 0.\overline{3}), (``2", 1.0 + 0.\overline{3}), (``3", 0.\overline{3}) \} \;.$$
%
%

\fi


\subsection{GroupBy}\label{sec:groupby}


The \texttt{GroupBy} operator implements functionality similar to MapReduce: it takes a key selector function (``mapper'') and a result selector function (``reducer"), and transforms a dataset first into a collection of groups, determined by the key selector, and then applies the result selector to each group, finally outputting pairs of key and result. This transformation is used to process groups of records, for example, collecting edges by their source vertex, so that the out-degree of the vertex can be determined. The traditional implementation of \texttt{GroupBy} (\eg in LINQ) is not stable, even on inputs with integral weights, because the presence or absence of a record in the input can cause one group in the output be replaced by a different group, corresponding to an addition and subtraction of an output record. This can be addressed by setting the output weight of each group to be \emph{half} the weight of the input records. This sufficient for most of the wPINQ examples in this paper, which only group unit-weight records.

\ifnum\full=0
We defer the more complex description of \texttt{GroupBy}'s behavior on general weighted datasets to our technical report.
\else

The general case of \texttt{GroupBy} on arbitrarily weighted inputs records is more complicated, but tractable.
%
%
Let the key selector function partition $A$ into multiple disjoint parts, so that $A = \sum_k A_k$. The result of wPINQ's \texttt{GroupBy} on $A$ will be the sum of its applications to each part of $A_k$:
$$ \texttt{GroupBy}(A, f) = \sum_k \texttt{GroupBy}(A_k) $$
%
%
%
For each part $A_k$, let $x_0, x_1, \ldots, x_i$ order records so that they are non-increasing in weight $A_k(x_i)$. For each $i$, \texttt{GroupBy} emits the set of the first $i$ elements in this order, with weight $(A_k(x_i) - A_k(x_{i+1}))/2$. That is,
$$ \texttt{GroupBy}(A_k)(\{ x_j : j \le i \}) = (A_k(x_i) - A_k(x_{i+1}))/2. $$
If all records have the same weight $w$, only the set of all elements has non-zero weight and it is equals to $w/2$.

As an example of \texttt{GroupBy}, consider grouping
\begin{eqnarray*}
C & = & \{ (``1", 0.75), (``2", 2.0), (``3", 1.0), (``4", 2.0), (``5", 2.0)\} \; .
\end{eqnarray*}
again using parity as a key selector. We produce four groups:
\begin{align*}
\{\;(\text{``odd, $\{5,3,1\}$''}, \; 0.375), \;
(\text{``odd, $\{5,3\}$''}, \; 0.125)\\
(\text{``odd, $\{5\}$ ''}, \; 0.5), \;
(\text{``even, $\{2,4\}$''}, \; 1.0)\;\}
\end{align*}

\noindent
Appendix~\ref{apx:ops} provides a proof of \texttt{GroupBy}'s stability.
\fi

\myparab{Node degree.} \texttt{GroupBy} can be used to obtain node degree:

{\small
\begin{verbatim}
//from (a,b) compute (a, da)
var degrees = edges.GroupBy(e => e.src, l => l.Count());
\end{verbatim}
}

\noindent
\texttt{GroupBy} groups edges by their source nodes, and then counts the number of edges in each group --- this count is exactly equal to the degree $d_a$ of the source node $a$ --- and outputs the record $\langle a,d_a \rangle$; since all input records have unit weight, the weight of each  output record is 0.5.

\subsection{Union, Intersect, Concat, and Except}\label{sec:concat}

wPINQ provides transformations that provide similar semantics to the familiar Union, Intersect, Concat, and Except database operators. \texttt{Union} and \texttt{Intersect} have weighted interpretations as element-wise min and max:
\begin{eqnarray*}
\texttt{Union}(A,B)(x) & = & \max(A(x), B(x))\\
\texttt{Intersect}(A,B)(x) & = & \min(A(x), B(x))
\end{eqnarray*}
\texttt{Concat} and \texttt{Except} can be viewed as element-wise addition and subtraction:
\begin{eqnarray*}
\texttt{Concat}(A,B)(x) & = & A(x) + B(x)\\
\texttt{Except}(A,B)(x) & = & A(x) - B(x)
\end{eqnarray*}
Each of these four are stable binary transformations.

\myparab{Example.}
Applying \texttt{Concat}$(A,B)$ with our sample datasets $A$ and $B$ we get
\begin{eqnarray*}
 \{ (``1", 3.75), (``2", 2.0),(``3", 1.0),(``4", 2.0) \} \; .
\end{eqnarray*}
and taking \texttt{Intersect}$(A,B)$ we get $\{ (``1", 0.75)\}$.

\subsection{The {Join} transformation.}\label{sec:join:wPINQ}

The \texttt{Join} operator is an excellent case study in the value of using weighted datasets, and the workhorse of our graph analysis queries.

Before describing wPINQ's \texttt{Join}, we discuss why the familiar SQL join operator fails to provide stability. The standard SQL relational equi-join takes two input datasets and a key selection function for each, and produces all pairs whose keys match. The transformation is not stable, because a single record in $A$ or $B$ could match as many as $\|B\|$ or $\|A\|$ records, and its presence or absence would cause the transformation's output to change by as many records~\cite{hayJoins}. To deal with this, the \texttt{Join} in PINQ suppresses matches that are not unique matches, damaging the output to gain stability, and providing little utility for graph analysis.

The wPINQ \texttt{Join} operator takes two weighted datasets, two key selection functions, and a reduction function to apply to each pair of records with equal keys. To determine the output of \texttt{Join} on two weighted datasets $A,B$, let $A_k$ and $B_k$ be their restrictions to those records mapping to a key $k$ under their key selection functions. The weight of each record output from a match under key $k$ will be scaled down by a factor of $\|A_k\| + \|B_k\|$. For the identity reduction function we can write this as
\begin{equation}\label{eq:joinweight}
\texttt{Join}(A, B) = \sum_k \frac{A_k\times B_k^T}{\|A_k\| + \|B_k\|}
\end{equation}
where the outer product $A_k \times B_k^T$ is the weighted collection of elements $(a,b)$ each with weight $A_k(a) \times B_k(b)$. The proof that this \texttt{Join} operator is stable appears in the Appendix~\ref{apx:ops}.

\myparab{Example.  } Consider applying \texttt{Join} to our example datasets $A$ and $B$ using ``parity" as the join key. Records with even and odd parity are
\begin{eqnarray*}
A_0 = \{ (``2", 2.0) \} & \textrm{ and } & A_1 = \{ (``1", 0.5), (``3", 1.0) \}\\
B_0 = \{ (``4", 2.0) \} & \textrm{ and } & B_1 = \{ (``1", 3.0) \}
\end{eqnarray*}
The norms are $\|A_0\| + \|B_0\| = 4.0$ and $\|A_1\| + \|B_1\| = 4.5$,
and scaling the outer products by these sums gives
\begin{eqnarray*}
A_0 \times B_0^T \; / \; 4.0 & = & \{ (``\langle2,4\rangle", 1.0) \} \\
A_1 \times B_1^T \; / \; 4.5 & = & \{ (``\langle1,1\rangle", 0.\overline{3}), (``\langle3,1\rangle", 0.\overline{6}) \}
\end{eqnarray*}
The final result is the accumulation of these two sets,
$$
\{ (``\langle2,4\rangle", 1.0), (``\langle1,1\rangle", 0.\overline{3}), (``\langle 3,1\rangle", 0.\overline{6}) \} \; .
$$

\myparab{Join and paths. } Properties of paths in a graph are an essential part of many graph analyses (Section~\ref{sec:graphs}). We use \texttt{Join} to compute the set of all paths of length-two in a graph, starting from the \texttt{edges} dataset of Section~\ref{sec:background}:

{\small
\begin{verbatim}
  //given edges (a,b) and (b,c) form paths (a,b,c)
  var paths = edges.Join(edges, x => x.dst, y => y.src,
                         (x,y) => new Path(x,y))
\end{verbatim}
}

\noindent
We join {\texttt{edges}} with itself, matching edges $(a,b)$ and $(b,c)$ using the key selectors corresponding to ``destination" and ``source" respectively.  Per (\ref{eq:joinweight}), the result is a collection of paths $(a,b,c)$ through the graph, each with weight $\frac1{2d_b}$, where $d_b$ is the degree of node $b$. Paths through high-degree nodes have smaller weight; this makes sense because each edge incident on such a vertex participates in many length two paths. At the same time, paths through low degree nodes maintain a non-trivial weight, so they can be more accurately represented in the output without harming their privacy. By diminishing the weight of each path proportionately, the influence of any one edge in the input is equally masked by a constant amount of noise in aggregation.

\subsection{Shave}\label{sec:shave}


Finally, we introduce an transformation that decomposes  records with large weight to be into multiple records with smaller weight. Such an operator is important in many analyses with non-uniform scaling of weight, to ensure that each output record has a common scale.


The \texttt{Shave}  transformation allows us to break up a record $x$ of weight $A(x)$ into multiple records $\langle x, i \rangle$ of smaller weights $w_i$ that sum to $A(x)$.   Specifically,  \texttt{Shave} takes in a function from records to a sequence of real values $f(x) = \langle w_0, w_1, \ldots \rangle$. For each record $x$ in the input $A$, \texttt{Shave} produces records $\langle x,0\rangle$, $\langle x,1\rangle$, $\ldots$ for as many terms as $\sum_i w_i \le A(x)$.  The weight of output record $\langle x,i\rangle$ is therefore
$$ \texttt{Shave}(A, f)(\langle x,i\rangle) = \max(0, \min(f(x)_i, A(x) - \sum_{j < i} f(x)_j)) \; . $$

\myparab{Example.}
Applying \texttt{Shave} to our sample dataset $A$ where we let $f(x)=\langle 1.0, 1.0, 1.0, ... \rangle$  $\forall x$, we obtain the dataset
$$\{ (``\langle 1,0\rangle ", 0.75), (``\langle 2,0\rangle ", 1.0), (``\langle 2,1 \rangle ", 1.0),
(``\langle 3,0\rangle ", 1.0), \} \;.$$
\texttt{Select} is \texttt{Shave}'s functional inverse; applying \texttt{Select} with $f(\langle x, i \rangle)=x$ that retains only the first element of the tuple, recovers the original dataset $A$ with no reduction in weight.

\myparab{Transforming edges to nodes. }  We now show how to use \texttt{SelectMany}, \texttt{Shave}, and \texttt{Where} to transform the {\texttt{edges}} dataset, where each record is a unit-weight edge, to a {\texttt{nodes}} dataset where each record is a node of weight 0.5.

{\small
\begin{verbatim}
var nodes = graph.SelectMany(e => new int[] { e.a, e.b })
                 .Shave(0.5)
                 .Where((i,x) => i == 0)
                 .Select((i,x) => x);
\end{verbatim}
}

\noindent
\texttt{SelectMany} transforms the dataset of edge records into a dataset of node records, \ie each unit-weight edge is transformed into two node records, each of weight 0.5. In wPINQ, the weights of identical records accumulate, so each node $x$ of degree $d_x$ has weight $\tfrac{d_x}{2}$.
Next, \texttt{Shave} is used convert each node record $x$ into a multiple records $\langle x,i\rangle$ for $i=0,...,d_x$, each with weight $0.5$.
\texttt{Where} keeps only the 0.5-weight record $\langle x,0\rangle$, and \texttt{Select} converts  record $\langle x,0\rangle$ into $x$.  The result is a dataset of nodes, each of weight $0.5$. Note that it is not possible to produce a collection of nodes with unit weight, because  each edge uniquely identifies two nodes; the presence or absence of one edge results in \emph{two} records changed in the output, so a weight of 0.5 is the largest we could reasonably expect from a stable transformation.


\section{Graph Analyses with \lowercase{w}PINQ}\label{sec:graphs}

In \cite{wosn} we showed how the built-in operators in wPINQ can be used for differentially-private analysis of first- and second-order graph statistics; namely, we considered degree and joint-degree distributions, and reproduced the work of \cite{hay} and \cite{sala} as wPINQ algorithms while correcting minor issues in \cite{hay,sala}'s analyses (the requirement that the number of nodes be public and a flaw in the privacy analysis, respectively. We provide a corrected privacy analysis in Appendix~\ref{apx:jddProof}.) These examples demonstrated that wPINQ is capable of expressing and bounding the privacy cost of challenging custom analyses.

We first start by reviewing the algorithms we defined for both degree distribution and joint-degree distribution. Then, we present new results that use wPINQ to compute more sophisticated third- and fourth-order graph statistics. These algorithms suggest new approaches for counting triangles and squares in differentially-private manner (Theorem~\ref{thm:tbd} and Theorem~\ref{thm:square}) which may be of independent interest.



%
\ifnum\full=1
\myparab{Data model. }
One can move to undirected edges and back using wPINQ operators (\texttt{Select} and \texttt{SelectMany}, respectively), resulting in weight $2.0$ for edges in the undirected graph. When comparing with prior work on undirected (assymetric) graphs, we provide equivalent privacy guarantees by doubling the amplitude of the added Laplace noise.
\fi

\subsection{Degree distributions}\label{sec:degreeDist}


Hay \etal~\cite{hay} show that the addition of Laplace noise to the non-decreasing sequence of vertex degrees provides differential privacy, if the number of nodes are publicly known. Because the original values are known to be non-decreasing, much of the noise can be filtered out by post-processing via isotonic regression (regression to an ordered sequence). 
However, requiring that the number of nodes is public makes it difficult for this approach to satisfy differential privacy; we could not apply the technique to a graph on cancer patients without revealing the exact number of patients, for example.

We show that Hay \etal's analysis can be reproduced in wPINQ without revealing the number of nodes, using a program that produces a non-{\em increasing} degree sequence rather than a non-decreasing degree sequence. The measured degree sequence continues indefinitely, with noise added to zeros, and it is up to the analyst to draw conclusions about where the sequence truly ends (perhaps by measuring the number of nodes using differential privacy). We also show how to make postprocessing more accurate by combining measurements with other measurements of the degree sequence, namely its complementary cumulative density function (CCDF).

\myparab{Degree CCDF. }  The degree  CCDF is the functional inverse of the (non-increasing) degree sequence; one can be obtained from the other by interchanging the $x$- and $y$-axes. We start with a wPINQ query for the degree CCDF, which we explain in detail below, and from which we build the degree sequence query:

{\small
\begin{verbatim}
  var degCCDF = edges.Select(edge => edge.src)
                     .Shave(1.0)
                     .Select(pair => pair.index);

  var ccdfCounts  = degCCDF.NoisyCount(epsilon);
\end{verbatim}
}
\noindent
The first step in the query is to apply the \texttt{Select} transformation, 
 which takes a function $f : D \rightarrow R$ and applies it to each input record, resulting in an accumulation of weighted records. Here, \texttt{Select} transforms the dataset from edges with weight 1.0 to a dataset of node names $a$, each weighted by $d_a$.

The next step applies the \texttt{Shave} transformation 
which transforms a dataset of weighted records $(``a", d_a)$ into a collection of indexed records with equal weight (the 1.0 specified as an argument):
$$\{ (``\langle a,0\rangle", 1.0), (``\langle a,1\rangle", 1.0), \ldots (``\langle a,d_a-1\rangle", 1.0) \} \; . $$
Finally, we apply \texttt{Select} to retain only the index $i$ of the pair $\langle a, i \rangle$, obtaining records $i=0,1,2, \ldots$. As wPINQ automatically accumulates weights of identical records, each record $i$ is weighted by the number of graph nodes with degree greater than $i$.
\texttt{NoisyCount} then returns a dictionary of  noised weights for each $i$.

\myparab{Degree Sequence.} We get the degree sequence by \emph{transposing} the $x$- and $y$-axis of the degree CCDF:

{\small
\begin{verbatim}
  var degSeq = degCCDF.Shave(1.0)
                      .Select(pair => pair.index);

  var degSeqCounts  = degSeq.NoisyCount(epsilon);
\end{verbatim}
}
\noindent
\texttt{Shave} and \texttt{Select} are as in the above example, just applied a second time. They transform the weighted set of records $i = 0, 1, \ldots$ (\texttt{degCCDF}) to a set of records $j = 0,1,\ldots$ where the weight of record $j$ is the number of input records $i$ with weight at least $j$, which is exactly the non-decreasing degree sequence.


\myparab{Postprocessing.}
We can think of the space of non-increasing degree sequences as drawing a path $P$ on the two-dimensional integer grid from position $(0, \infty)$ to position $(\infty,0)$ that only steps down or to the right.  Our goal is to find such a path $P$ that fits our noisy degree sequence and noisy CCDF  as well as possible. More concretely, thinking of $P$ as a set of cartesian points $(x,y)$, and given the noisy ``horizontal" ccdf measurements $h$ and the noisy ``vertical" degree sequence measurements $v$, we want to minimize
\begin{equation}\label{eq:minCost}
\sum_{(x,y)\in P}|v[x]-y|+|h[y]-x| \; .
\end{equation}
To do this, we weight directed edges in the two-dimensional integer grid as
\begin{eqnarray*}
cost((x,y) \rightarrow (x+1,y)) & = & |v[x] - y| \\
cost((x,y+1) \rightarrow (x,y)) & = & |h[y] - x| \; .
\end{eqnarray*}
and compute the lowest cost path from $(0,N)$ to $(N,0)$ for some large $N$. To see how this works, notice that the cost of taking a horizontal step from $(x,y)$ to $(x+1,y)$ is a function of the ``vertical'' degree sequence measurement $v$; we do this because taking a horizontal step essentially means that we are committing to the vertical value $y$.  (And vice versa for the vertical step.) Thus, finding the lowest-cost path allows us to simultaneously fit the noisy CCDF and degree sequence while minimizing (\ref{eq:minCost}).

\ifnum\full=1
Computing this lowest-cost path computation is more expensive than the linear time PAVA computation. However, our algorithm constructs edges only as needed, and the computation only visits nodes in the low cost ``trough" near the true measurements, which reduces complexity of computation. Experimentally, the computation takes several milliseconds.
\fi

\subsection{Joint degree distribution}\label{sec:jdd}

Sala \etal~\cite{sala} investigate the problem of reporting the JDD, \ie number of edges incident on vertices of degree $d_a$ and $d_b$. After non-trivial analysis, \ifnum\full=1(reproduced by us in Appendix~\ref{apx:jddProof}), \fi they determine that adding Laplace noise with parameter $4\max(d_a, d_b) / \epsilon$ to each count provides $\epsilon$-differential privacy.
This result bypasses worst-case sensitivity, because number of nodes $|V|$ and maximum degree $d_{\max}$ do not appear in the noise amplitude. However, \cite{sala}'s approach had a very subtle flaw in its privacy analysis: the count for $\langle d_a,d_b\rangle$ was released \emph{without} noise if there was no edge in the graph that was incident with nodes of degrees $d_a$ and $d_b$. If this is the case for many pairs $\langle d_a, d_b\rangle$, it can reveal a large amount of information about the secret graph. While \cite{sala}'s approach can be corrected by releasing noised values for all pairs $\langle d_a, d_b\rangle$, doing this would harm the accuracy of their experimental results.

We now write a wPINQ query that is similar in spirit to the analysis of Sala \etal, trading a constant factor in accuracy for an automatic proof of privacy (as well as some automatic improvement in Section~\ref{sec:synth}):


\begin{small}
\begin{verbatim}
  // (a, da) for each vertex a
  var degs = edges.GroupBy(e => e.src, l => l.Count());

  // ((a,b), da) for each edge (a,b)
  var temp = degs.Join(edges, d => d.key, e => e.src,
                       (d,e) => new { e, d.val });

  // (da, db) for each edge (a,b)
  var jdd = temp.Join(temp,
                      a => a.edge,
                      b => b.edge.Reverse(),
                      (a,b) => new { a.val, b.val });

  var jddCount = jdd.NoisyCount(epsilon);
\end{verbatim}
\end{small}

\noindent
The first step uses wPINQ's \texttt{GroupBy} transformation, which takes in a key selector function (here \texttt{e => e.src}) and a reducer function (here \texttt{l => l.Count()}), and transforms a dataset into a collection of groups, one group for each observed key, containing those records mapping to the key. The reducer function is then applied to each group.  While \texttt{GroupBy} uses a more complex weight rescaling when input records have unequal weights, 
 when input records all have equal weight (here $1.0$) the weight of output records is simply halved (to $0.5$).
Thus, each element of \texttt{degs} is a 0.5-weight (key, val) pair, where the key $a$ is the name of a (source) vertex, and its value is its degree $d_a$.

%

To obtain the joint degree distribution we use \texttt{Join} twice, first to combine \texttt{edges} with \texttt{degs}, and then second to bring $d_a$ and $d_b$ together for each edge $(a,b)$. After the first \texttt{Join}, the dataset \texttt{temp} contains records of the form $\langle\langle a,b\rangle,d_a\rangle$. Per (\ref{eq:joinweight})  each has weight
$$ \frac{0.5 \times 1.0}{0.5 + d_a} = \frac{1}{1 + 2d_a} \; .$$
Finally, we \texttt{Join}  \texttt{temp} with itself, to match pairs of the form $\langle\langle a,b\rangle,d_a\rangle $ with $\langle \langle b,a\rangle,d_b\rangle $ using key $\langle a,b\rangle$.  There is exactly one match for each output record $\langle d_a,d_b\rangle$ so applying (\ref{eq:joinweight}) gives weight
\begin{equation}\label{eq:jdd:weight}
\frac{\frac{1}{1 + 2d_a} \times \frac{1}{1 + 2d_b}}{\frac{1}{1 + 2d_a} + \frac{1}{1 + 2d_b}}   =  \frac{1}{2 + 2d_a + 2d_b} \; .
\end{equation}

\noindent
If we multiply the results of \texttt{NoisyCount} for record $\langle d_a, d_b\rangle $ by the reciprocal of \eqref{eq:jdd:weight} we get a noised count with error proportion to  $2 + 2d_a + 2d_b$.
\ifnum\full=1
Although this appears better than Sala \etal's result, this number doesn't tell the full story.
Our approach produces measurements for both $\langle d_a, d_b\rangle$ and $\langle d_b, d_a\rangle$, giving twice the weight if we combine them before measuring with \texttt{NoisyCount}. At the same time, Sala \etal use undirected graphs, giving twice the privacy guarantee for the same value of epsilon, canceling the previous improvement.
\fi
We used the input dataset four times, so by the remark in Section~\ref{sec:trans} our privacy cost is $4\epsilon$.  Therefore, to give a comparable privacy bound to Sala~\etal, we should add noise of amplitude $8 + 8d_a + 8d_b$ to the true count in each $\langle d_a,d_b\rangle$ pair, a result that is worse than Sala \etal's bespoke analysis by a factor of between two and four.


\subsection{Triangles by degree (TbD)}\label{sec:tribydeg}

We now work through an example wPINQ analysis to count the number of triangles incident on vertices of degrees $d_a$, $d_b$, and $d_c$, for each triple $\langle d_a, d_b, d_c\rangle$. This statistic is useful for the graph generation model of~\cite{sigcomm06}. The idea behind our algorithm is that if a triangle $\langle a,b,c\rangle$ exists in the graph, then paths $abc$, $cab$, and $bca$ must exist as well.  The algorithm forms each path $abc$ and pairs it with the degree of its internal node, $d_b$. The paths are then joined together, each time with a different rotation, so that $abc$ matches $cab$ and $bca$, identifying each triangle and supplying the degrees of each incident vertex.

The first action turns the undirected set of edges into a symmetric directed graph, by concatenating \texttt{edges} with its transpose.

\begin{small}
\begin{verbatim}
  var edges = edges.Select(x => new { x.dst, x.src })
                   .Concat(edges);
\end{verbatim}
\end{small}

\noindent
The implication of this transformation is that each subsequent use of \texttt{edges} will in fact correspond to two uses of the undirected source data.

We now compute length-two paths $\langle a, b, c \rangle$ by joining the set of edges with itself. We then use the \texttt{Where} transformation to discard length-two cycles (paths of the form $\langle a, b, a \rangle$).

\begin{small}
\begin{verbatim}
  var paths = edges.Join(edges, x => x.dst, y => y.src,
                         (x,y) => new Path(x, y))
                   .Where(p => p.a != p.c);
\end{verbatim}
\end{small}
\noindent
As in the example of Section~\ref{sec:join:wPINQ} the weight of $\langle a, b, c \rangle$ is $\tfrac1{2d_b}$.

We next determine the degree of each vertex with a \texttt{GroupBy}, producing pairs $\langle v, d_v \rangle$ of vertex and degree, as in the example in Section~\ref{sec:groupby}. We join the result with the path triples $\langle a, b, c\rangle$, producing pairs of path and degree $\langle \langle a ,b, c \rangle, d_b\rangle$.

\begin{small}
\begin{verbatim}
  var degs = edges.GroupBy(e => e.src, l => l.Count());
  var abc = paths.Join(degs, abc => abc.b, d => d.key,
                        (abc,d) => new { abc, d.val });
\end{verbatim}
\end{small}

\noindent
The pairs \texttt{degs} produced by \texttt{GroupBy} each have weight $1/2$, and joining them with \texttt{paths} results in records with weight
$$ \frac{\frac{1}{2} \times \frac{1}{2d_b}}{\frac{1}{2} + {d_b(d_b-1) \times \frac{1}{2d_b}}} = \frac{1}{2d_b^2} \; ,$$
per equation~(\ref{eq:joinweight}). The number of length-two paths through $b$ is $d_b(d_b-1)$ rather than $d_b^2$ because we discarded cycles.

We next use \texttt{Select} to create two rotations of $\langle \langle a ,b, c \rangle, d_b\rangle$, namely $\langle \langle b, c, a \rangle, d_b\rangle$ and $\langle \langle c, a, b \rangle, d_b\rangle$. The resulting path-degree pairs hold the degrees of the first and third vertices on the path, respectively. Because they were produced with \texttt{Select}, their weights are unchanged.

\begin{small}
\begin{verbatim}
  var bca = abc.Select(x => { rotate(x.path), x.deg });
  var cab = bca.Select(x => { rotate(x.path), x.deg });
\end{verbatim}
\end{small}

\noindent
We join each of the permutations using the length-two path as the key, retaining only the triple of degrees $\langle \langle d_a,d_b,d_c \rangle$.

\begin{small}
\begin{verbatim}
  var tris = abc.Join(bca, x => x.path, y => y.path,
                      (x,y) => new { x.path, x.deg, y.deg })
                .Join(cab, x => x.path, y => y.path,
                      (x,y) => new { y.deg, x.deg1, x.deg2 });
\end{verbatim}
\end{small}

\noindent
Each \texttt{Join} will be a unique match, since each path $\langle b, c, a \rangle$ occurs only once in the dataset. The weight of records in the output of the first \texttt{Join} is therefore
$$ \frac{\frac{1}{2d_b^2}\times \frac{1}{2d_c^2}}{\frac{1}{2d_b^2} + \frac{1}{2d_c^2}}  = \frac{1}{2d_b^2 + 2d_c^2} \; .$$
By the same reasoning, the weight of records in the output of the second \texttt{Join} (\ie the \texttt{tris} datasets), reflecting all three rotations and degrees, will be
\begin{equation}\label{eq:triByDegScale}
\frac1{2(d_a^2 + d_b^2 + d_c^2)} \; .
\end{equation}

Finally, we sort each degree triple so that all six permutations of the degree triple $\langle d_a,d_b,d_c \rangle$ all result in the same record. We use \texttt{NoisyCount} to measure the total weight associated with each degree triple.

\begin{small}
\begin{verbatim}
  var order = triangles.Select(degrees => sort(degrees));
  var output = order.NoisyCount(epsilon);
\end{verbatim}
\end{small}

\noindent
Each triangle $\langle a, b, c \rangle$ contributes weight $1/2(d_a^2 + d_b^2 + d_c^2)$ six times, increasing the weight of $\langle d_a, d_b, d_c \rangle$ by $3/(d_a^2 + d_b^2 + d_c^2)$. Nothing other than triangles contribute weight to any degree triple.
Dividing each reported number by (\ref{eq:triByDegScale}) yields the number of triangles with degrees $d_a, d_b, d_c$, plus noise.

This query uses the input dataset \texttt{edges} eighteen times (three permutations, each using \texttt{edges} three times, doubled due to the conversion (with \texttt{Concat}) to a symmetric directed edge set). To get $\epsilon$-DP we must add Laplace noise with parameter $18/\epsilon$ to each count. Divided by $3/(d_a^2 + d_b^2 + d_c^2)$, the error associated with a triple $\langle d_a, d_b, d_c \rangle$ is a Laplace random variable with parameter $6(d_a^2 + d_b^2 + d_c^2)/\epsilon$. Thus,

\begin{theorem}\label{thm:tbd}
For an undirected graph $G$, let $\Delta(x,y,z)$ be the number of triangles in $G$ incident on vertices of degree $x, y, z$.
The mechanism that releases
$$ \Delta(x,y,z) + \textrm{Laplace}(6 (x^2 + y^2 + z^2) / \epsilon) $$
 for all $x,y,z$ satisfies $\epsilon$-differential privacy.
\end{theorem}

\noindent
Section~\ref{sec:synth} discusses our experiments with this query.

\subsection{Squares by degree (SbD)}


Approaches similar to those used to count triangles can be used to design other subgraph-counting queries.  To illustrate this, we present a new algorithm for counting squares (\ie cycles of length four) in a graph.  The idea here is that if a square $abcd$ exists in the graph, then paths $abcd$, $cadb$ must exist too, so we find these paths and \texttt{Join} them together to discover a square.

The first three steps of the algorithm are identical to those in TbD; we again obtain the collection \texttt{abc} of length-two paths $abc$ along with the degree $d_b$. Next, we join \texttt{abc} with itself, matching paths $abc$ with paths $bcd$, to obtain length-three paths $abcd$ with degrees $d_b$ and $d_c$. We use the \texttt{Where} operator to discard cycles (\ie paths $abca$).

{\small
\begin{verbatim}
var abcd = abc.Join(abc, x => x.bc, y => y.ab, 
                     (x, y) => 
                     { new Path(x.bc,y.ab), y.db, x.db })
              .Where(y => y.abcd.a != y.abcd.d);
\end{verbatim}
}

\noindent
We then use \texttt{Select} to rotate the paths in \texttt{abcd} twice:
{\small
\begin{verbatim}
var cdab = abcd.Select(x => 
             { rotate(rotate(x.abcd)), x.db, x.dc });
\end{verbatim}
}

\noindent
If a square $abcd$ exists in the graph, then record $(abcd, d_d, d_a)$ will be in the rotated set \texttt{cdab}.  We therefore join \texttt{abcd} with \texttt{cdab}, using the path as the joining key, to collect all four degrees $d_a,d_b,d_c,d_d$. Sorting the degrees coalesces the eight occurrences of each square (four rotations in each direction).

{\small
\begin{verbatim}
var squares = abcd.Join(cdab, x => x.abcd, y => y.abcd,
                (x, y) => new { y.da, x.db, x.dc, y.dd });

var order  = squares.Select(degrees => sort(degrees));
var output = order.NoisyCount(epsilon);
\end{verbatim}
}


\ifnum\full=1 

\myparab{Analysis of output weights.} As in the TbD, we obtain the tuples $\langle a,b,c, db \rangle$ with weight $1 \over 2d_{b}^{2}$. Then we \texttt{Join} \texttt{abc$d_b$} with itself 
to obtain length-three paths passing through the edge (b,c); there are $(d_b - 1)(d_c -1)$ length-three paths passing through this edge, of which $d_b -1$ have weight $1 \over 2d_{b}^{2}$ and $d_c -1$ have weight $1 \over 2d_{c}^{2}$. Each record in \texttt{abcd} therefore has weight
\begin{equation}
{{{1 \over 2d_{b}^{2}} \times {1 \over 2d_{c}^{2}}} \over {(d_b -1){1 \over 2d_b^2} + (d_c -1){1 \over 2d_c^2}}} = {1 \over {2(d_b^2(d_c-1) + d_c^2(d_b-1))}}
\end{equation}
In the final \texttt{Join}, each record matches at most one other record (every path $\langle a,b,c,d \rangle$ matches exactly one path $\langle c,d,a,b \rangle$ ). The weight of the each record in \texttt{squares} is thus:
\begin{equation}
 1 \over {2(d_a^2(d_d-1)+d_d^2(d_a-1)+d_b^2(d_c-1)+d_c^2(d_b -1))}
\end{equation}
The final step orders the quadruples $\langle d_a,d_b,d_c,d_d \rangle$ which increases the weight of each square incident the quadruple by a factor of $8$, because each square is discovered eight times (\ie all rotations of $a,b,c,b$ and all rotations of $d,c,b,a$).

This time, the input dataset is used 12 times, so the privacy cost is $12\epsilon$.   Adding the an additional factor of two to account for using directed graphs rather than undirected graphs, we obtain the following new result:

\else

An analysis of the algorithm (deferred to our technical report) leads to the following new result:

\fi

\begin{theorem}\label{thm:square}
For an undirected graph $G$, let $\Box(v,x,y,z)$ be the number of cycles of length 4 in $G$ incident on vertices of degree $v,x,y,z$.
The mechanism that releases
$$ \Box(v,x,y,z) + \textrm{Laplace}(6 (vx(v + x) + yz(y + z)) / \epsilon) $$
for all $v,x,y,z$ satisfies $\epsilon$-differential privacy.
\end{theorem}

\subsection{Counting arbitrary motifs}

Motifs are small subgraphs, like triangles or squares, whose prevalence in graphs can indicate sociological phenomena. The approach we have taken, forming paths and then repeatedly \texttt{Join}ing them to tease out the appropriate graph structure, can be generalized to arbitrary connected subgraphs on $k$ vertices. However, the analyses in this section were carefully constructed so that all records with the same degrees had exactly the same weight, allowing us to measure them separately and exactly interpret the meaning of each measurement. More general queries, including those for motifs, combine many records with varying weights, complicating the interpretation of the results. Fortunately, we introduce techniques to address this issue in the next section.

%
%

\section{Synthesizing input datasets}
\label{sec:pi}

DP queries can produce measurements that are not especially accurate, that exhibit inherent inconsistencies (due to noise), or that may not be directly useful for assessing other statistics of interest.
One approach to these problems is the use of \emph{probabilistic inference}~\cite{pi}, in which the precise probabilistic relationship between the secret dataset and the observed measurements is used, via Bayes' rule, to produce a posterior distribution over possible datasets. The posterior distribution integrates all available information about the secret dataset in a consistent form, and allows us to sample synthetic datasets on which we can evaluate arbitrary functions, even those without privacy guarantees.

While previous use of probabilistic inference (see \eg~\cite{pi}) required human interpretation of the mathematics and custom implementation, wPINQ automatically converts all queries into an efficient Markov chain Monte Carlo (MCMC) based sampling algorithm. This automatic conversion is accomplished by an incremental data-parallel dataflow execution engine which, having computed and recorded the measurements required by the analyst, allows MCMC to efficiently explore the space of input datasets by repeatedly applying and evaluating small changes to synthetic inputs. Note that while this process does use the noisy wPINQ measurements, it no longer uses the secret input dataset; all synthetic datasets are public and guided only by their fit to the released differentially private wPINQ measurements.

\subsection{Probabilistic Inference}

The main property required for a principled use of probabilistic inference is an exact probabilistic relationship between the unknown input (\eg the secret graph) and the observed measurements (\eg the noisy answers to wPINQ queries). Although the input is unknown, we can draw inferences about it from the measurements if we know how likely each possible input is to produce an observed measurement.

For each query $Q$, dataset $A$, and measured observation $m$, there is a well-specified probability, $\Pr[Q(A) + Noise = m]$, that describes how likely each dataset $A$ is to have produced an observation $m$ under query $Q$. For example, when adding Laplace noise with parameter $\epsilon$ to multiple counts defined by $Q(A)$ (Section~\ref{sec:agg}), that probability is
\begin{equation}\label{eq:pMgivenA}
 \Pr[Q(A) + Noise = m] \propto \exp(\epsilon \times \|Q(A) - m\|_1) \; .
\end{equation}
While this probability is notationally more complicated when different values of $\epsilon$ are used for different parts of the query, its still easily evaluated.

This probability 
informs us about the relative likelihood that dataset $A$ is in fact the unknown input dataset. Bayes' rule shows how these probabilities update a prior distribution to define a posterior distribution over datasets, conditioned on observations $m$:
\begin{align*}
 \Pr[A | m] &= \Pr[m|A] \times \frac{\Pr[A] }{ \Pr[m]} \; .
\end{align*}
$\Pr[m|A]$ is the probability $m$ results from $Q(A) + N$, as in (\ref{eq:pMgivenA}). $\Pr[m]$ is a normalizing term independent of $A$. $\Pr[A]$ reflects any prior distribution over datasets. The result is
\begin{align*}
 \Pr[A | m] & \propto \exp(\epsilon \times \|Q(A) - m\|_1) \times \Pr[A]\; .
\end{align*}
The constant of proportionality does not depend on $A$, so we can use the right hand side to compare the relative probabilities $\Pr[A|m]$ and $\Pr[A'|m]$ for two datasets $A$ and $A'$.

The posterior distribution focuses probability mass on datasets that most closely match the observed measurements, and indirectly on datasets matching other statistics that these measurements constrain.
For example, while Theorem~\ref{thm:tbd} indicates that DP measurement of the TbD requires significant noise to be added to high-degree tuples, the posterior distribution combines information from the TbD with highly accurate DP measurements of degree distribution~\cite{wosn,hay} to focus on graphs respecting both, effectively downplaying TdD measurements resulting primarily of noise and likely improving the fit to the total number of triangles.

However, the posterior distribution is a mathematical object, and we must still compute it, or an approximation to it, before we achieve its desireable properties.

\newpage
\subsection{Metropolis-Hastings}\label{sec:metrohasting}

The posterior probability $\Pr[A|m]$ is over all datasets, an intractably large set. Rather than enumerate the probability for all datasets $A$, modern statistical approaches use sampling to draw representative datasets.
Metropolis-Hastings is an MCMC algorithm that starts from a supplied prior  over datasets, and combines a random walk over these datasets with a function scoring datasets, to produce a new random walk whose limiting distribution is proportional to the supplied scores. Eliding important (but satisfied) techincal assumptions, namely that the random walk be reversible and with a known statitionary distibution, the pseudo-code is:

\vspace{-2mm}
{\small
\begin{verbatim}
  var state;   // initial state

  while (true)
  {
     // random walk proposes new state
     var next = RandomWalk(state);

     // change state with probability Min(1, newScore/old)
     if (random.NextDouble() < Score(next) / Score(state))
        state = next;
  }
\end{verbatim}
}

\noindent
The user  provides an initial value for \texttt{state}, a random walk \texttt{RandomWalk}, and a scoring function \texttt{Score}.

\myparab{Choosing the initial state.} MCMC is meant to converge from any starting point, but it can converge to good answers faster from well-chosen starting points. Although we can start from a uniformly random dataset, we often seed the computation with a random dataset respecting some observed statistics. In the case of graphs, for example, we choose a random `seed' graph matching wPINQ measurements of the secret graph's degree distribution (Section~\ref{sec:workflow}).

\myparab{Choosing the random walk. }
We allow the user to specify the random walk, though a natural default is to replace a randomly chosen element of the collection with another element chosen at random from the domain of all possible input records.  Our random walk for graphs is in Section~\ref{sec:workflow}.
More advanced random walks exist, but elided assumptions (about reversibility, and easy computation of relative stationary probability) emerge as important details.

\myparab{Choosing the scoring function. }  As we are interested in a distribution over inputs that better matches the observed measurements, we will take as scoring function
$$ Score_{\langle Q,m\rangle}(A) = \exp(\epsilon \times \|Q(A) - m\|_1 \times pow) \; . $$
Our initial distribution over states is uniform, allowing us to to discard the prior distribution $\Pr[A]$ in the score function, and so when $pow = 1$ this score function results in a distribution proportional to $\exp(\epsilon \times \|Q(A) - m\|_1)$, proportional to the posterior distribution suggested by probabilistic inference.
The parameter $pow$ allows us to focus the output distribution, and make MCMC behave more like a greedy search for a single synthetic dataset. Large values of $pow$ slow down the convergence of MCMC, but eventually result in outputs that more closely fit the measurements $m$.

\subsection{Incremental query evaluator}\label{sec:imp}

In each iteration, MCMC must evaluate a scoring function on a small change to the candidate dataset. This scoring function essentially evaluates a wPINQ query on the changed data, which can be an expensive computation to perform once, and yet we will want to perform it repeatedly at high speeds. We therefore implement each wPINQ query as an \emph{incrementally updateable computation}, using techniques from the view maintenance literature. Each MCMC iteration can thus proceed in the time it takes to incrementally update the computation in response to the proposed small change to the candidate dataset, rather than the time it takes to do the computation from scratch. We now describe this incremental implementation and its limitations.

To incrementally update a computation, one needs a description of how parts of the computation depend on the inputs, and each other. This is often done with a directed dataflow graph, in which each vertex corresponds to a transformation and each edge indicates the use of one transformation's output as another transformation's input.
As an analyst frames queries, wPINQ records their transformations in such a directed dataflow graph. Once the query is evaluated on the first candidate dataset, a slight change to the input (as made in each iteration of MCMC) can propagate through the acyclic dataflow graph, until they ultimately arrive at the \texttt{NoisyCount} endpoints and incrementally update $\| Q(A)  - m \|_1$ and the score function.


Each wPINQ transformation must be implemented to respond quickly to small changes it its input. Fortunately, all of wPINQ's transformations are \emph{data-parallel}.
A transformation is data-parallel if it can be described as a single transformation applied independently across a partitioning of its inputs. For example, \texttt{Join} is data-parallel because each of its inputs is first partitioned by key, and each part is then independently processed in an identical manner. Data-parallelism is at the heart of the stability for wPINQ's transformations (Section~\ref{sec:trans}):  a few changed input records only change the output of their associated parts. Data-parallelism also leads to efficient incremental implementations, where each transformation can maintain its inputs indexed by part, and only recomputes the output of parts that have changed. As these parts are typically very fine grained (\eg an individual \texttt{Join} key), very little work can be done to incrementally update transformations; outputs produced from keys whose inputs have not changed do not need to be reprocessed.
All of wPINQ's transformations are data-parallel, and are either be implemented as a stateless pipeline operators (\eg \texttt{Select}, \texttt{Where}, \texttt{Concat}) or a stateful operators whose state is indexed by key (\eg \texttt{GroupBy}, \texttt{Intersect}, \texttt{Join}, \texttt{Shave}).
\ifnum\full=0
The details of our incremental implementations of wPINQ's transformations are deferred to our technical report. 

\else
Most of wPINQ's transformations have implementations based on standard incremental updating patterns, maintaining their inputs indexed by key so that changes to inputs can quickly be combined with existing inputs in order to determine the difference between ``before" and ``after" outputs. In several cases the implementations can be optimized; a standard relational \texttt{Join}, for example, uses distributivity to incrementally update its output without explicitly subtracting the old output from the new output:
$$ \overbrace{(A + a) \bowtie (B + b)}^\textrm{new output} = \overbrace{A \bowtie B}^\textrm{old output} + \overbrace{a \bowtie B + A \bowtie b + a \bowtie b}^\textrm{incremental update} \; . $$
In contrast, wPINQ's \texttt{Join} may need to change the weights of all output records with the same key if the normalization in equation (\ref{eq:joinweight}) changes, still far fewer records than in the whole output.
However, in cases where the sum of the input weights for a key are unchanged (\eg when inputs change from one value to another but maintain the same weight), wPINQ's \texttt{Join} is optimized to only perform as much work as the relational \texttt{Join}.  More details are in Appendix~\ref{sec:opsImplement}.
\fi

Two limitations to wPINQ's MCMC performance are computation time and memory capacity. Incrementally updating a computation takes time dependent on the number of weights that change; in complex graph analyses, a single changed edge can propagate into many changed intermediate results. At the same time,  wPINQ maintains state in order to incrementally update a computation; this state also scales with the size of these intermediate results, which grow with the size of the input data size and with query complexity.
In the TbD from Section~\ref{sec:tribydeg}, for example, one edge $(a,b)$ participates in up to $O(d_a^2 + d_b^2)$ candidate triangles, each of which may need to have its weight updated when  $(a,b)$'s weight changes. Also, the final \texttt{Join}
operators in TbD match arbitrary length-two paths, which wPINQ stores indexed in memory; the memory required for the TbD therefore scales as $\sum_v d_v^2$ (the number of length-two paths) which can be much larger than number of edges $\sum_v d_v$.

Depending on the complexity of the query, wPINQ's MCMC iterations can complete in as few as $50$ microseconds, though typical times are closer to hundreds of milliseconds on complex triangle queries. MCMC can require arbitrarily large amounts of memory as the size of the data grows; we have tested it on queries and datasets requiring as many as $64$ gigabytes of memory. Distributed low-latency dataflow systems present an opportunity to scale wPINQ to clusters with larger aggregate memory. 
More detail on running time and memory footprint are reported in Section~\ref{sec:tbi}.

We plan to publicly release the implementation once we have properly cleaned and documented the code.

\section{Experiments}\label{sec:synth}

We now apply our wPINQ's query and probabilistic inference to the workflow for graph synthesis proposed in~\cite{wosn}, extended with further measurements. We start by using wPINQ to produce noisy DP-measurements of a secret graph and then use MCMC to synthesize graphs respecting these noisy measurements. Our approach is similar to that of \cite{sigcomm06}, who identify and measure key graph statistics which constrain broader graph properties.  However, while \cite{sigcomm06} starts from \emph{exact} degree correlation measurements, we start from DP measurements of degree correlation that can be noisy and inconsistent.
In~\cite{wosn}, we presented several positive results of applying our platform to this problem,  generating synthetic graphs that respect degree distribution and/or joint degree distribution of a private graph. We do not reproduce these results here; instead, we present new results for the more challenging problem of counting triangles.

We start by presenting experiments on symmetric directed graphs with where the total privacy cost is a constant times $\epsilon$ where $\epsilon = 0.1$ and MCMC parameter $pow=10,000$ (see Section~\ref{sec:metrohasting}).  We investigate the sensitivity of our results to different values of $\epsilon$ in Section~\ref{sec:tbi}, and then close the loop on the discussion in Section~\ref{sec:imp} by experimenting with the scalability of our platform.


\subsection{A workflow for graph synthesis (from \cite{wosn})}\label{sec:workflow}

To provide necessary background, we briefly reproduce the description of the workflow for graph synthesis that we sketched in~\cite{wosn}.  The workflow begins with the analyst's wPINQ queries to the protected graph. 
Once the queries are executed and noisy measurements are obtained, the protected graph is discarded.  Graph synthesis proceeds, using \emph{only the noisy measurements}, as follows:

\myparab{Phase 1. Create a ``seed'' synthetic graph.}  In~\cite{wosn} we showed wPINQ queries and regression techniques that result in a highly accurate $\epsilon$-differentially private degree sequence. We then seed a simple graph generator that generates a random graph fitting the measured   $\epsilon$-DP degree sequence.  This random graph is the initial state of our MCMC process.

\myparab{Phase 2. MCMC. }  The synthetic graph is then fit to the wPINQ measurements, using MCMC to search for graphs that best fit the measurements. Starting from our seed graph, we  use an {edge-swapping random walk} that preserves the degree distribution of the seed synthetic graph; at every iteration of MCMC, 
we propose replacing two random edges $(a,b)$ and $(c,d)$  with edges $(a,d)$ and $(c,b)$. As MCMC proceeds, the graph evolves to better fit the  measurements.

\begin{table}
\begin{center}
\begin{scriptsize}
\begin{tabular}{|l||r|r|r|r|r|}
  \hline
  Graph & Nodes & Edges &  $d_{\max}$ & $\triangle$ & $r$ \\
 \hline
  CA-GrQc   & 5,242 & 28,980 & 81 & 48,260 &  0.66 \\
  CA-HepPh & 12,008 & 237,010 & 491 & 3,358,499  & 0.63  \\
  CA-HepTh & 9,877 & 51,971 & 65 & 28,339  & 0.27\\
  Caltech & 769 & 33,312 & 248 & 119,563 & -0.06 \\
    Epinions & 75,879 & 1,017,674 & 3,079  & 1,624,481  & -0.01   \\
     Random(GrQc) & 5,242 & 28,992 & 81 & 586 & 0.00\\
  Random(HepPh) & 11,996 & 237,190 & 504 & 323,867  & 0.04\\
  Random(HepTh) & 9,870 & 52,056 & 66 & 322 & 0.05\\
  Random(Caltech) & 771 & 33,368 & 238 & 50,269  & 0.17 \\
  Random(Epinion) & 75,882 & 1,018,060 & 3,085  &  1,059,864 &  0.00  \\

  \hline
\end{tabular}
\end{scriptsize}
\vspace{-3mm}
\caption{Graph statistics:  number of triangles ($\Delta$), assortativity ($r$), and maximum degree $d_{\max}$.
}
\label{tab:graphStats}
\vspace{-8mm}
\end{center}
\end{table}


\subsection{Evaluating Triangles by Degree (TbD)}\label{sec:exp:tribydeg}

%


Our goal is now to combine MCMC with measurements of third-order degree correlations in order to draw inferences about graph properties we are not able to directly measure: specifically, assortativity $r$ and the number of triangles $\triangle$. We find that it can be difficult to exactly reconstruct detailed statistics like the number of triangles with specified degrees, but that these measurements nonetheless provide information about aggregate quantities like $r$ and $\triangle$.

We start by considering generating synthetic graphs using our triangles by degree (TbD) query described in Section \ref{sec:tribydeg}.   While this query has appealing bounds for small degrees, it still requires each $\langle d_a,d_b,d_c\rangle$ triple to have its count perturbed by noise proportional to $6({d_a^2 + d_b^2 + d_c^2})/\epsilon$. For large degrees, many of which will not occur in the secret graph, the results are almost entirely noise. Unable to distinguish signal from noise, MCMC makes little  progress.
%

Fortunately, the flexibility of our workflow allows for a simple remedy to this problem: instead of counting each degree triple individually, we first group triples into buckets with larger cumulative weight. The noise added to each bucket remains the same, but the weight (signal) in each bucket increases, which enables MCMC to distinguish graphs with different numbers of triangles at similar degrees.

%
The following modification to the code in Section~\ref{sec:tribydeg} replaces each degree by the floor of the degree divided by $k$, effectively grouping each batch of $k$ degrees into a bucket:

\begin{small}
\begin{verbatim}
 var degs = edges.GroupBy(e => e.src, l => l.Count()/k);
\end{verbatim}
\end{small}

\noindent
MCMC will automatically aim for synthetic graphs whose bucketed degree triples align with the measurements, without any further instruction about how the degrees relate.

\myparab{Experiments.} We experiment with our workflow using the graph CA-GrQc~\cite{leskovec2007graph} (see statistics in Table~\ref{tab:graphStats}). We first generate the `seed' synthetic graph that is fit to wPINQ measurements of (a) degree sequence, (b) degree complementary cumulative density function, and (c) count of number of nodes (see~\cite{wosn} for details); the privacy cost of generating the seed synthetic graph is $3\epsilon=0.3$.  MCMC then fits the seed graph to TbD measurements with privacy cost $9\epsilon=0.9$ for $\epsilon=0.1$, so the total privacy cost of this analysis is $0.9+0.3=1.2$.

Figure~\ref{fig:tbd} shows MCMC's progress (after $5 \times 10^{6}$ steps), plotting the number of triangles $\triangle$ and the assortativity $r$ in the synthetic graph as MCMC proceeds. Figure~\ref{fig:tbd} also plots MCMC's behavior when the secret graph is ``Random(GrQc)'', a random graph with the same degree distribution as CA-GrQc but with very few triangles, used to see if MCMC can distinguish the two.
The figures reveal that MCMC is only able to distinguish the two graphs when bucketing is used ($k = 20$), but still does not find graphs respecting the input graph properties. This is likely due to a lack of signal in the result of the TbD query, which may be compounded by the approximate nature of MCMC or the restricted random walk we have chosen for it to use.

\begin{figure}[t]
\begin{center}
\vspace{-4mm}
\hspace*{-0.2in}
\includegraphics[width=0.5\textwidth]{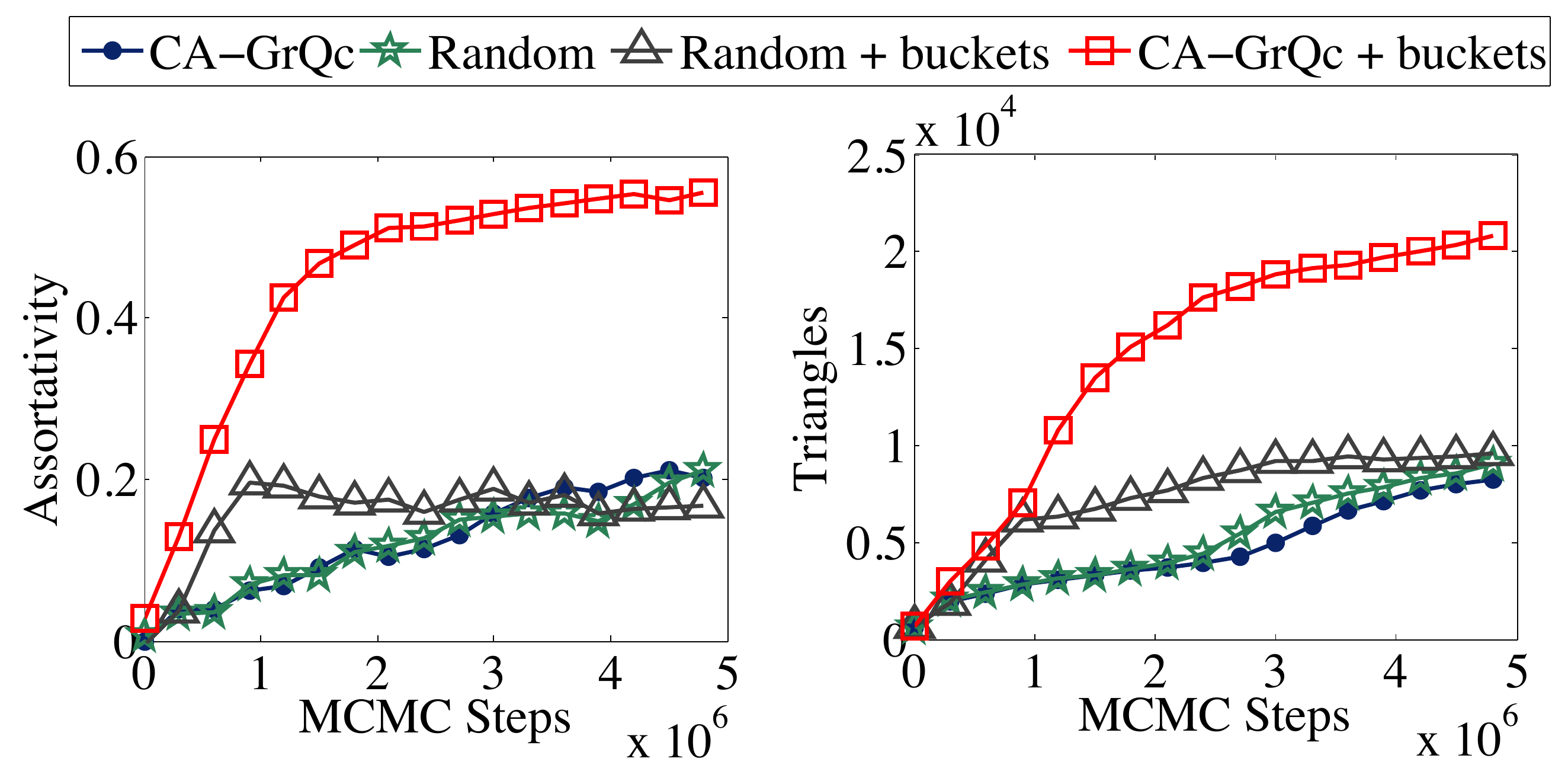}
\vspace{-8mm}
\caption{Behavior of the TbD query with and without bucketing for CA-GrQc. The values on CA-GrQc are $r = 0.66$, $\triangle = 48,260$ and on the sanity check $r = 0.0$, $ \triangle = 586$. Bucketing substantially improves accuracy.}\label{fig:tbd}\vspace{-5mm}
\end{center}
\end{figure}

\ifnum\full=1
Indeed, we found that for $\epsilon=0.2$ (which is what we used in our measurements), the available signal in the TdB is below the noise level for almost every bucket apart from the lowest degree one. Indeed, the total weight available in the TbD for the true GrQc graph is only $\approx89$, and when we bucket by 20, we find that $83\%$ of that signal is concentrated in the lowest-degree bucket. Meanwhile, Laplace noise with parameter $\epsilon=0.2$ has amplitude is $5$, so its not surprising that the signal is dwarfed by the noise almost everywhere. The limited signal means that (with the exception of the lowest-degree bucket) MCMC will find a synthetic graph that is fitting the noise in the TbD, rather than any useful information about triangles from the protected graph GrQc. Moreover, as the number of signal-free buckets increases, we have more opportunities for the zero-mean Laplace noise in that bucket to ``mislead'' MCMC into thinking there are triangles where there are none.
\fi

\subsection{Evaluating Triangles by Intersect (TbI)}\label{sec:tbi}


The TbD query of Section~\ref{sec:tribydeg} has the property that the measurements
%
are relatively easy to interpret, but this does not neccesarily translate into good performance. We now consider an alternate query, whose results are harder to interpret directly, but that gives much better results when run through our MCMC workflow.
%
%
The query, Triangles By Intersect (TbI) measures a single quantity related to the number of total triangles, reducing the amount of noise introduced. The algorithm is again based on the observation a triangle $abc$ exists if and only if both of the length-two paths $abc$ and $bca$ exist. This time, however, we create the collection of length two paths, and \emph{intersect} it with itself after appropriately permuting each path:

\begin{small}
\begin{verbatim}
  // form paths (abc) for a != c with weight 1/db
  var paths = edges.Join(edges, x => x.dst, y => y.src,
                         (x,y) => new Path(x, y))
                   .Where(p => p.a != p.c);

  //rotate paths and intersect for triangles
  var triangles = paths.Select(x => rotate(x))
                       .Intersect(paths);

  //count triangles all together
  var result =  triangles.Select(y => "triangle!")
                         .NoisyCount(epsilon);
\end{verbatim}
\end{small}

\ifnum\full=0
Deferring the analysis of weights to our technical report, we state the value TbI produces before wPINQ applies noise:
\else
The analysis of weights is as follows.
As in Section~\ref{sec:tribydeg}, \texttt{Join} creates length-two  paths \texttt{abc} each with weight ${1 \over 2 d_b}$. Because the graph is symmetric (\ie both edges $(a,b)$ and $(b,a)$ are present), a single triangle $(a,b,c)$ will contribute six records to the \texttt{abc} collection; paths $(b,a,c)$ and $(c,a,b)$ each with weight $\tfrac1{2d_a}$, two paths with weight $\tfrac1{2d_b}$, and two paths with $\tfrac1{2d_c}$.   Permuting the paths \texttt{abc} dataset to obtain \texttt{bca} using \texttt{Select} does not alter record weights, and since \texttt{Intersect} takes the minimum weight of the participating elements,  the dataset \texttt{triangles} now contains two records with weight $\min\{\tfrac1{2d_a},\tfrac1{2d_b}\}$, two with weight $\min\{\tfrac1{2d_a},\tfrac1{2d_c}\}$, and two with weight $\min\{\tfrac1{2d_b},\tfrac1{2d_c}\}$ for each triangle $(a,b,c)$ in the graph.  Finally \texttt{Select} aggregates all the weight to obtain a \emph{single} count with weight
\fi
\begin{equation}\label{eq:tbiWeight}
\sum_{\Delta(a,b,c)} \min \{\tfrac1{d_a}, \tfrac1{d_b}\}
+\min \{\tfrac1{d_a}, \tfrac1{d_c}\}
+\min \{\tfrac1{d_b}, \tfrac1{d_c}\}
\end{equation}
TbI uses the input dataset four times, incurring a privacy cost of $4\epsilon$, which is less than the $9\epsilon$ incurred by TbD.

Notice that the TbI outputs only a single noised count, rather than counts by degree. While this single count might be difficult for a human to interpret, it contains information about triangles, and MCMC will work towards finding a synthetic graph that fits it.


\begin{figure}[t]
\begin{center}
 \vspace{-3mm}
\includegraphics[width=0.5\textwidth,height=2.6in]{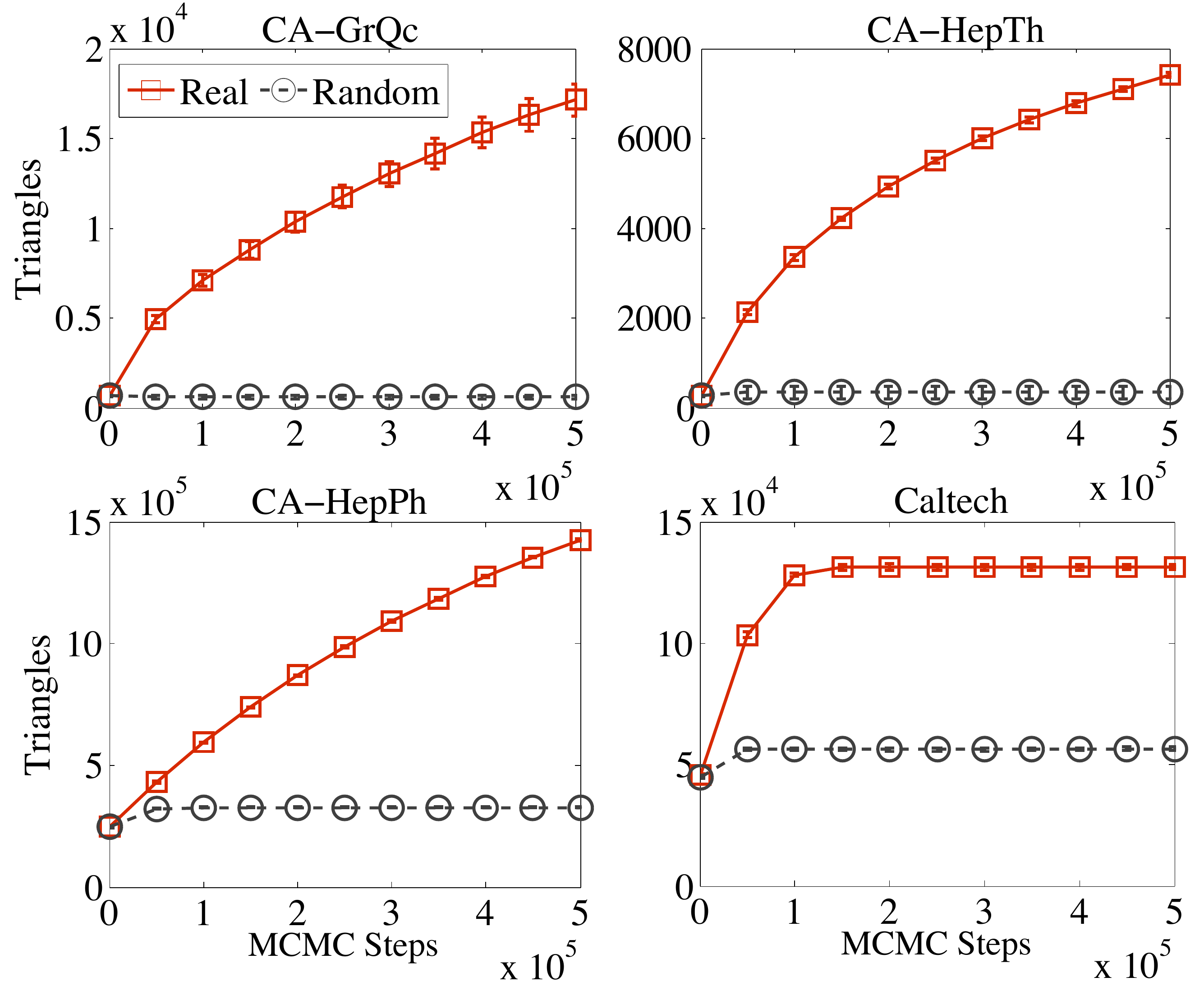}
\vspace{-8mm}
\caption{Fitting the number of triangles with TbI, $5 \times 10^5$ MCMC steps. }\label{fig:triangles}\vspace{-5mm}
\end{center}
\end{figure}

\begin{table}[t]
\begin{center}
\begin{scriptsize}
    \begin{tabular}{|c|c|c|c|c|c|}
    \hline
    ~ & ~                & CA-GrQc & CA-HepPh & CA-HepTh & Caltech   \\\hline

   ~    & Seed &  643 & 248,629 & 222 & 45,170 \\
   $\Delta$             & MCMC  &  35,201 & 2,723,633 & 16,889 & 129,475 \\
~               & Truth &  48,260 & 3,358,499 & 28,339 & 119,563 \\\hline     
    \end{tabular}
    \end{scriptsize}
     \vspace{-2mm}
  \caption { $\Delta$ before MCMC, after $5 \times 10^6$ MCMC steps using TbI, and in the original graph.}\label{tab:relErrorTris}
  \vspace{-6mm}
    \end{center}
\end{table}

\myparab{Experiments.} In Figure~\ref{fig:triangles}, we plot the number of triangles $\Delta$ versus the number of iterations of MCMC ($5 \times 10^5$ steps) for synthetic graphs generated from our workflow, on both actual and random graphs. We generated the seed graphs as in Section~\ref{sec:exp:tribydeg} with privacy cost $3\epsilon=0.3$. The TbI query has privacy cost $4\epsilon=0.4$, and so the total privacy cost is $7\epsilon=0.7$.  We see a clear distinction between real graphs and random graphs; MCMC introduces triangles for the real graphs as appropriate, and does not for the random graphs.
Table~\ref{tab:relErrorTris} reports the initial, final, and actual number of triangles for an order of magnitude more MCMC steps.

There is still room to improve the accuracy of triangle measurement using these techniques, but our results do show that even very simple measurements like TbI, which provide little direct information about the number of triangles, can combine with degree distribution information to give non-trivial insight into quantities that are not easily measured directly. wPINQ allows us to experiment with new queries, automatically incorporating the implications of new measurements without requiring new privacy analyses for each new query.

\myparab{Different values of $\epsilon$.} We repeated the previous experiment with different values of $\epsilon\in \{0.01,0.1,1,10\}$ (for total privacy cost $7\epsilon$).  For brevity, we present results for the CA-GrQc graph and the corresponding random graph Random(GrQc) in Figure~\ref{fig:multi-eps}.
We see that the choice of $\epsilon$ does not significantly impact the behavior of MCMC, maintaining roughly the same expected value but with increases in variance for larger values of $\epsilon$ (\ie noisier queries).
MCMC remains well-behaved because the ``signal'' in the TbI query over the GrQC graph (\ie the value of equation~(\ref{eq:tbiWeight})) is large enough to be distinguished from both random noise and the signal in a random graph.



\begin{figure}[b]
\begin{center}
 \vspace{-3mm}
\includegraphics[height=1.5in, width=0.4\textwidth]{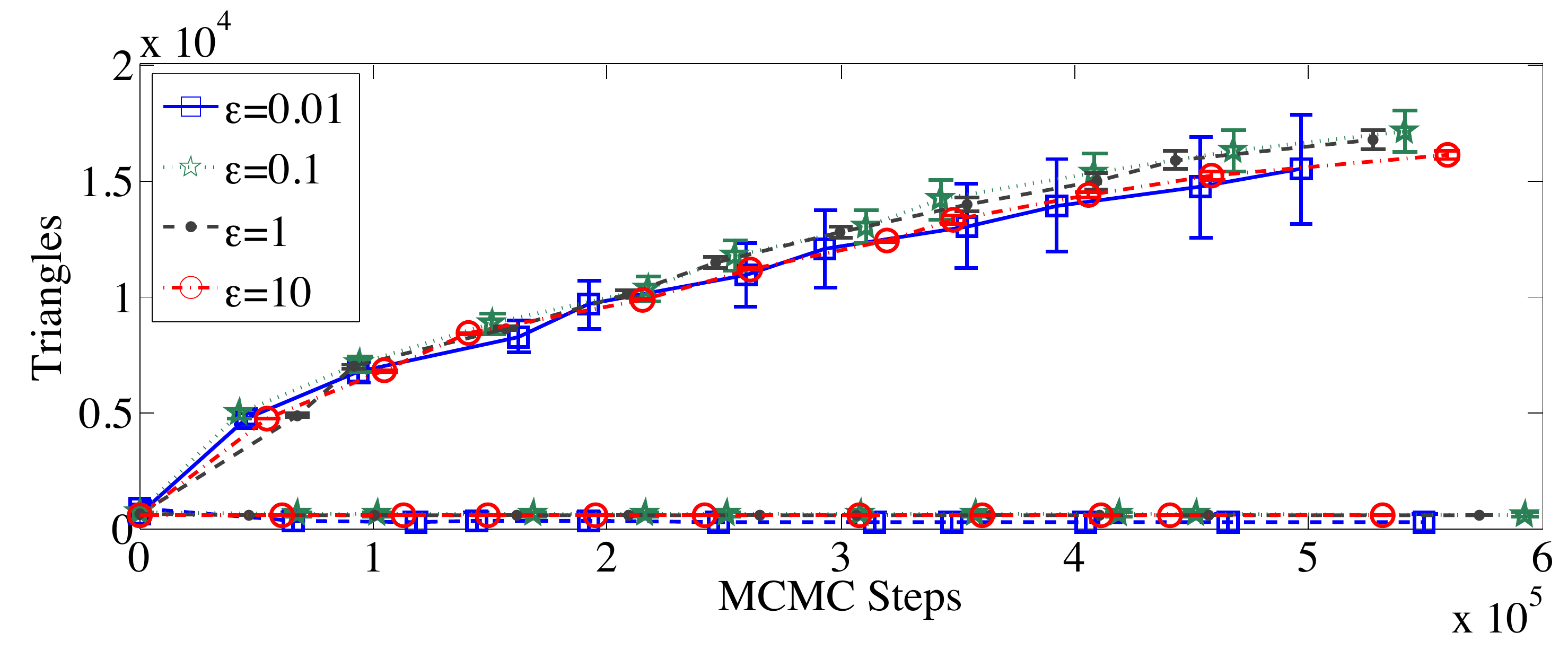}
\vspace{-5mm}
\caption{Testing TbI with different values of $\epsilon$; error bars show standard deviation based on 5 repeated experiments. Total privacy cost for each query is $7\epsilon$. The position of the measurements on the x-axis have been randomized for better visualization.}
\label{fig:multi-eps}
\vspace{-5mm}
\end{center}
\end{figure}

\myparab{Scalability analysis.} We now experiment with TbI using larger datasets, to provide insight into the running time and memory requirements of our approach. The memory requirements should grow with $O(\sum_{v \in G} d_v^2)$, as these are the number of candidate triangles (Section~\ref{sec:imp}). The running time should increase with the skew in the degree distribution, as each edge is incident on more candidate triangles.

%

To verify the scaling properties of TbI, we use five synthetic graphs drawn from the Barab\`{a}si-Albert distribution.  Barab\`{a}si-Albert graphs follow a power law degree distribution (similar to some social networks) and the generation process uses a preferential attachment model. We fix the number of nodes at $100,000$ and edges at $2M$ and change the degree of highest-degree nodes by increasing ``dynamical exponent'' of the preferential attachment~\cite{barabasi2013network} as $\beta \in \{0.5,0.55,0.6,0.65,0.7\}$. \ifnum\full=1 The resulting graphs are described in Table~\ref{tab:barabasi}.\fi
Experimental results are in Figure~\ref{fig:ds-tbi}. As $\sum_{v \in G} d_v^2$ increases from about 71M (for the graph with $\beta=0.5$) to 119M (for $\beta=0.7$), the memory required for MCMC increases.  Meanwhile computation rate, \ie MCMC steps/second,  decreases. Our platform can perform approximately 80 MCMC steps/second on the ``easiest'' ($\beta=0.5$) graph using about 25Gb of RAM. For the most difficult graph ($\beta=0.7$), it can perform about 25 MCMC steps/second using over 45Gb of RAM.

We conclude the scalability analysis with a performance evaluation of TbI on Epinions~\cite{richardson2003trust}, see Table~\ref{tab:graphStats}. While Epinions has about half the number of edges as the most difficult Barab\`{a}si-Albert graphs we tried ($\beta=0.7$), the quantity $\sum_{v \in G} d_v^2=224$M for Epinions is almost double that of the Barab\`{a}si-Albert graph, making Epinions the most difficult graph we tried on our platform. To work with Epinions, we needed over 50GB of memory, and computation ran no faster than 10 MCMC steps/second.  As usual, we experimented with this graph and a random graph with the same distribution but less triangles (statistics in Table~\ref{tab:graphStats}). We run the MCMC for $100,000$ steps and compute the number of triangles every $10,000$ steps (Figure~\ref{fig:ds-tbi}).

\ifnum\full=1
\begin{table}[t]
{\centering
\begin{scriptsize}
    \begin{tabular}{|l||c|c|c|c|c|}
    \hline
    \textbf{Graph} & \textbf{Nodes}   & \textbf{Edges}     & \textbf{$d_{\max}$} & \textbf{$\Delta$} & $\sum_{v \in G} d_v^2$ \\ \hline
    Barab\`{a}si 1 & 100,000 & 2,000,000 & 377        & 16,091 & 71,859,718   \\
    Barab\`{a}si 2 & 100,000 & 2,000,000 & 475        & 18,515 & 77,819,452  \\
    Barab\`{a}si 3 & 100,000 & 2,000,000 & 573        & 22,209 & 86,576,336  \\
    Barab\`{a}si 4 & 100,000 & 2,000,000 & 751        & 28,241 & 99,641,108   \\
    Barab\`{a}si 5 & 100,000 & 2,000,000 & 965        & 35,741 & 119,340,328  \\ \hline
    \end{tabular}
    \end{scriptsize}
}\vspace{-3mm}
     \caption{Statistics of the Barab\`{a}si-Albert graphs.}
       \label{tab:barabasi}\vspace{-4mm}
\end{table}
\fi

\begin{figure}[t]
\begin{center}
\includegraphics[width=0.5\textwidth]{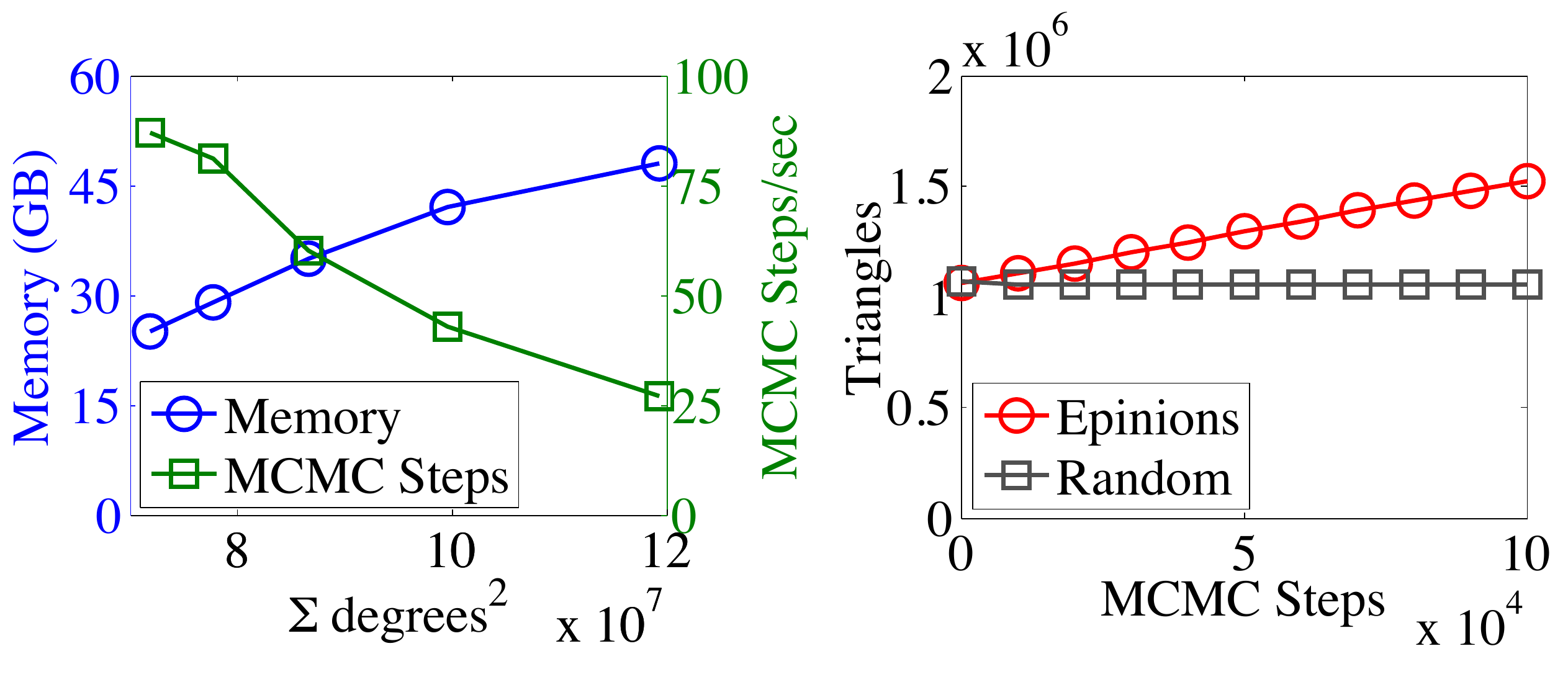}
\vspace{-20pt}
\caption{Running time and MCMC step/second for TbI computed on Barab\`{a}si-Albert graphs with 100K nodes, 2M edges, and dynamical exponent of $\beta\in\{0.5,0.55,0.6,0.65,0.7\}$ (left). TbI behavior on the epinions graph (right).
}
\vspace{-20pt}
\label{fig:ds-tbi}
\end{center}
\end{figure}

\section{Related work}\label{sec:related}

\myparab{Bypassing worse case sensitivity. }
Since the introduction of differential privacy~\cite{DMNS06}, there have been several approaches bypassing worst-case sensitivity~\cite{smooth,chopping,nodeprivacy,hayJoins}.  In the smooth sensitivity framework~\cite{smooth}, one adds noise determined from the sensitivity of the {specific input dataset} 
and those datasets near it. While these approaches can provide very accurate measurements, they typically require custom analyses and can still require large noise to be added even if only a few records in the dataset lead to worst-case sensitivity.  Weighted datasets allow us to down-weight the {specific} records in the dataset that lead to high sensitivity, leaving the benign component relatively untouched.

New approaches \cite{chopping,nodeprivacy} have surmounted sensitivity challenges by discarding portions of the input dataset that cause the sensitivity to be too high; for example, node-level differential privacy can be achieved by trimming the graph to one with a certain maximum degree and then using worst-case sensitivity bounds.  Our approach can be seen as a smoother version of this; we scale down the weights of portions of the input dataset, instead of discarding them outright.


An alternate approach relaxes the privacy guarantee for the portions of the input dataset that cause sensitivity to be high. Investigating joins in social graphs, Rastogi \etal~\cite{hayJoins} consider a relaxation of differential privacy in which more information release is permitted for higher degree vertices. Our approach can be seen as making the opposite compromise, sacrificing accuracy guarantees for such records rather than privacy guarantees.

In \cite{matrixAlg1,matrixAlg2} the authors use weighted sums with non-uniform weights to optimize collections of linear (summation) queries. While this is a popular class of queries, their techniques do not seem to apply to more general problems. With weighted datasets we can design more general transformations (\eg \texttt{Join}, \texttt{GroupBy}) that are crucial for graph analysis but not supported by \cite{matrixAlg1,matrixAlg2}.


\myparab{Languages. } Languages for differentially private computation started with PINQ~\cite{PINQ}, and have continued through Airavat~\cite{airavat}, Fuzz~\cite{fire,DFuzz}, and GUPT~\cite{gupt}. To the best of our knowledge, none of these systems support data-dependent rescaling of record weights.  Although Fuzz does draw a uniform rescaling operator (the ``\texttt{!}" operator) from work of Reed and Pierce~\cite{reed2010distance}, the programmer is required to specify a \emph{uniform} scaling constant for \emph{all} records (and is essentially equivalent to scaling up noise magnitudes).


\myparab{Privacy and graphs. }
Bespoke analyses for graph queries that provide edge-level differential privacy have recently emerged, including degree distributions~\cite{hay}, joint degree distribution (and assortativity)~\cite{sala}, triangle counting~\cite{smooth}, generalizations of triangles~\cite{KRSYgraphs}, and clustering coefficient~\cite{cluster}. New results have emerged for node-level differential privacy as well~\cite{nodeprivacy,chopping,recursive}. Of course, any wPINQ analysis can be derived from first principles; our contribution over these approaches is not in enlarging the space of differentially private computation, but rather in automating proofs of privacy and the extraction of information.

\cite{privacyTut} covers other graph analyses that satisfy privacy definitions that may not exhibit the robustness of DP.

\myparab{Bibliographic note.}  As we mentioned throughout, our earlier workshop paper~\cite{wosn} sketched our workflow and presented preliminary results showing how it could be used to synthesize graphs that respect degree and joint-degree distributions.  This paper is full treatment of our platform, showing how weighted transformation stability can form the basis of an expressive declarative programming language wPINQ, and presenting our incremental query processing engine for probabilistic inference.  We also present new algorithms and experiments related to counting triangles and squares.

\section{Conclusions}

We have presented our platform for differentially-private computation that consists of a declarative programming language, wPINQ, and an incremental evaluation engine that enables Markov chain Monte Carlo (MCMC) methods to synthesize representative datasets.  wPINQ is based on an approach to differentially-private computation, where data, rather than noise, is calibrated to the sensitivity of query. Specifically, wPINQ works with \emph{weighted datasets} so that the contribution of specific troublesome records that can harm privacy (\eg edges incident on high-degree nodes) are smoothly scaled down.   We have specialized our platform to private analysis of social graphs, and discussed how it can simplify the process, both by automating the proofs of privacy and the extraction of information needed for generating synthetic graphs.  While we have cast a number of analyses as queries in wPINQ and evaluated their performance, the analyses we have shown here are by no means the limit of what is possible with wPINQ. Indeed, we believe wPINQ's key benefit is its flexibility, and we therefore hope our platform will be an enabler for future private analyses of interesting datasets, especially social networks.

\myparab{Acknowledgements. } We thank the WOSN'12 reviewers, the VLDB'14 reviewers, George Kollios, Robert Lychev and Evimaria Terzi for comments on this draft, and Ran Canetti and Leonid Reyzin for helpful discussions.

\newpage
\begin{scriptsize}
\bibliographystyle{abbrv}
\bibliography{DP}

\begin{thebibliography}{10}

\bibitem{barabasi2013network}
A.-L. Barab{\'a}si.
\newblock Network science.
\newblock {\em Philosophical Transactions of the Royal Society A: Mathematical,
  Physical and Engineering Sciences}, 371(1987), 2013.

\bibitem{chopping}
J.~Blocki, A.~Blum, A.~Datta, and O.~Sheffet.
\newblock Differentially private data analysis of social networks via
  restricted sensitivity.
\newblock In {\em Proceedings of the 4th conference on Innovations in
  Theoretical Computer Science}, pages 87--96. ACM, 2013.

\bibitem{recursive}
S.~Chen and S.~Zhou.
\newblock Recursive mechanism: towards node differential privacy and
  unrestricted joins.
\newblock In {\em SIGMOD'13}, pages 653--664. ACM, 2013.

\bibitem{DMNS06}
C.~Dwork, F.~McSherry, K.~Nissim, and A.~Smith.
\newblock Calibrating noise to sensitivity in private data analysis.
\newblock In {\em Theory of Cryptography}, volume 3876 of {\em Lecture Notes in
  Computer Science}, pages 265--284. 2006.

\bibitem{DFuzz}
M.~Gaboardi, A.~Haeberlen, J.~Hsu, A.~Narayan, and B.~C. Pierce.
\newblock Linear dependent types for differential privacy.
\newblock In {\em {ACM} {SIGPLAN--SIGACT} {S}ymposium on {P}rinciples of
  {P}rogramming {L}anguages ({POPL}), Rome, Italy}, Jan. 2013.

\bibitem{fire}
A.~Haeberlen, B.~C. Pierce, and A.~Narayan.
\newblock Differential privacy under fire.
\newblock In {\em Proceedings of the 20th USENIX Security Symposium}, Aug.
  2011.

\bibitem{hay}
M.~Hay, C.~Li, G.~Miklau, and D.~Jensen.
\newblock Accurate estimation of the degree distribution of private networks.
\newblock In {\em IEEE ICDM '09}, pages 169 --178, dec. 2009.

\bibitem{privacyTut}
M.~Hay, K.~Liu, G.~Miklau, J.~Pei, and E.~Terzi.
\newblock Tutorial on privacy-aware data management in information networks.
\newblock Proc. SIGMOD'11, 2011.

\bibitem{KRSYgraphs}
V.~Karwa, S.~Raskhodnikova, A.~Smith, and G.~Yaroslavtsev.
\newblock Private analysis of graph structure.
\newblock In {\em Proc. {VLDB}'11}, pages 1146--1157, 2011.

\bibitem{nodeprivacy}
S.~Kasiviswanathan, K.~Nissim, S.~Raskhodnikova, and A.~Smith.
\newblock Graph analysis with node-level differential privacy, 2012.

\bibitem{leskovec2007graph}
J.~Leskovec, J.~Kleinberg, and C.~Faloutsos.
\newblock Graph evolution: Densification and shrinking diameters.
\newblock {\em ACM Transactions on Knowledge Discovery from Data (TKDD)},
  1(1):2, 2007.

\bibitem{matrixAlg1}
C.~Li, M.~Hay, V.~Rastogi, G.~Miklau, and A.~McGregor.
\newblock Optimizing linear counting queries under differential privacy.
\newblock pages 123--134, 2010.

\bibitem{matrixAlg2}
C.~Li and G.~Miklau.
\newblock An adaptive mechanism for accurate query answering under differential
  privacy.
\newblock pages 514--525, 2012.

\bibitem{sigcomm06}
P.~Mahadevan, D.~Krioukov, K.~Fall, and A.~Vahdat.
\newblock Systematic topology analysis and generation using degree
  correlations.
\newblock In {\em SIGCOMM '06}, pages 135--146, 2006.

\bibitem{MT}
F.~McSherry and K.~Talwar.
\newblock Mechanism design via differential privacy.
\newblock In {\em Proceedings of the 48th Annual IEEE Symposium on Foundations
  of Computer Science}, FOCS '07, pages 94--103, 2007.

\bibitem{PINQ}
F.~D. McSherry.
\newblock Privacy integrated queries: an extensible platform for
  privacy-preserving data analysis.
\newblock In {\em Proc. SIGMOD '09}, pages 19--30, 2009.

\bibitem{gupt}
P.~Mohan, A.~Thakurta, E.~Shi, D.~Song, and D.~Culler.
\newblock Gupt: privacy preserving data analysis made easy.
\newblock In {\em Proceedings of the 2012 international conference on
  Management of Data}, pages 349--360. ACM, 2012.

\bibitem{smooth}
K.~Nissim, S.~Raskhodnikova, and A.~Smith.
\newblock Smooth sensitivity and sampling in private data analysis.
\newblock In {\em {ACM STOC} '07}, pages 75--84, 2007.

\bibitem{wosn}
D.~Proserpio, S.~Goldberg, and F.~McSherry.
\newblock A workflow for differentially-private graph synthesis.
\newblock In {\em Proceedings of the 2012 ACM workshop on Workshop on online
  social networks}, pages 13--18. ACM, 2012.

\bibitem{hayJoins}
V.~Rastogi, M.~Hay, G.~Miklau, and D.~Suciu.
\newblock Relationship privacy: output perturbation for queries with joins.
\newblock In {\em Proc. PODS '09}, pages 107--116, 2009.

\bibitem{reed2010distance}
J.~Reed and B.~Pierce.
\newblock Distance makes the types grow stronger: A calculus for differential
  privacy.
\newblock {\em ACM Sigplan Notices}, 45(9):157--168, 2010.

\bibitem{richardson2003trust}
M.~Richardson, R.~Agrawal, and P.~Domingos.
\newblock Trust management for the semantic web.
\newblock In {\em The Semantic Web-ISWC 2003}, pages 351--368. Springer, 2003.

\bibitem{airavat}
I.~Roy, S.~T. Setty, A.~Kilzer, V.~Shmatikov, and E.~Witchel.
\newblock Airavat: security and privacy for mapreduce.
\newblock In {\em Proc. USENIX NSDI'10}, pages 20--20. USENIX Association,
  2010.

\bibitem{sala}
A.~Sala, X.~Zhao, C.~Wilson, H.~Zheng, and B.~Y. Zhao.
\newblock Sharing graphs using differentially private graph models.
\newblock In {\em IMC}, 2011.

\bibitem{cluster}
Y.~Wang, X.~Wu, J.~Zhu, and Y.~Xiang.
\newblock On learning cluster coefficient of private networks.
\newblock In {\em Proceedings of the IEEE/ACM International Conference on
  Advances in Social Networks Analysis and Mining (ASONAM)}.

\bibitem{pi}
O.~Williams and F.~McSherry.
\newblock Probabilistic inference and differential privacy.
\newblock {\em Proc. NIPS}, 2010.

\end{thebibliography}
\end{scriptsize}

\appendix

\section{Stability of \lowercase{w}PINQ operators}
\label{apx:ops}

We prove the stability of \texttt{Join} discussed in Section~\ref{sec:join:wPINQ}.
\begin{theorem}
For any datasets $A,A',B,B'$,
\begin{eqnarray*}
\|{\texttt{\em Join}}(A,B) - \texttt{\em Join}(A',B')\| & \le & \| A - A' \|  + \| B - B'\| \; .
\end{eqnarray*}
\end{theorem}

\begin{proof}
We will first argue that
\begin{equation}\label{eq:join1}
\|{\texttt{Join}}(A,B) - \texttt{Join}(A',B)\| \le \|A - A'\| \;.
\end{equation}
An equivalent argument shows  $\|{\texttt{Join}}(A',B) - \texttt{Join}(A',B')\| \le \|B - B'\|$ and concludes the proof.

It suffices to prove (\ref{eq:join1}) for any $A_k, A'_k, B_k$. Writing \texttt{Join} in vector notation as
$$
\texttt{Join}(A,B) = \sum_k \frac{A_k \times B_k^T}{\|A_k\| + \|B_k\|}
$$
we want to show that for each term in this sum,
$$
\left\| \frac{A_k \times B_k^T}{\|A_k\| + \|B_k\|} - \frac{A'_k \times B_k^T}{\|A'_k\| + \|B_k\|}\right\|  \le \|A_k - A'_k\| \; .
$$
The proof is essentially by cross-multiplication of the denominators and tasteful simplification.
For simplicity, let $a$ and $b$ be $A_k \times B_k^T$ and $A'_k \times B_k^T$, respectively, and let $x$ and $y$ be the corresponding denominators.  We apply the equality
$$\frac{a}{x} - \frac{b}{y} = \frac{a(y - x) - (b - a)x}{xy}, $$
followed by the triangle inequality:
$$\|\frac{a}{x} - \frac{b}{y}\| \le \frac{\|a(y - x)\|}{xy} + \frac{\|(b - a)x\|}{xy}. $$
Expanding, these two numerators can be re-written as
\begin{eqnarray*}
\|a(y - x)\|& = & (\|A_k\| - \|A'_k\|) \times (\|A_k\| \|B_k\|) \\
& \le &  \|A_k - A'_k\| \times (\|A_k\| \|B_k\|)\\
\|(b-a)x\| & = & \|(A_k - A'_k)B_k\| \times (\|A_k\| + \|B_k\|) \\
&\le &  \|A_k - A'_k\| \times \|B_k\| (\|A_k\| + \|B_k\|)
\end{eqnarray*}
Assuming, without loss of generality, that $\|A_k\| \ge \|A'_k\|$, their sum is at most
\begin{eqnarray*}
&& \|A_k - A'_k\| \times (\|A_k\| + \|A'_k\| + \|B_k\|) \times \|B_k\| \\
& = & \|A_k - A'_k\| \times (xy - \|A_k\|\|A'_k\|) \; .
\end{eqnarray*}
Division by $xy$ results in at most $\|A_k - A'_k\|$.
For our choice of $a,b,x,y$ the term $\|b-a\|x$ has a factor of $|x-y|$ in it, which we extract from both terms.
If we re-introduce the definitions of $a,b,x,y$ and apply a substantial amount of simplification, this bound becomes
$$
\|A_k - A'_k\| \times \frac{(\|A_k\| + \|A'_k\| + \|B_k\|) \times \|B_k\|}{ (\|A_k\| + \|B_k\|) \times (\|A'_k\| + \|B_k\|)}\; .
$$
The numerator is exactly $2\|A_k\|\|A'_k\|$ less than the denominator, making the fraction at most one.
\end{proof}

\ifnum\full=1

\noindent
Next, we prove stability for \texttt{GroupBy} (Section~\ref{sec:groupby}).

\begin{theorem}
For any $A,A'$ and key function $f$,
$$ \|\texttt{\em GroupBy}(A,f) - \texttt{\em GroupBy}(A',f) \| \le \|A - A'\| $$
\end{theorem}

\begin{proof}
We decompose each difference $A_k-A_k'$ into a sequence of differences, each to a single record and by an amount $\delta$ that does not change the order of records within the group. That is, each $\delta$ satisfies
For each $x_i$, a change in weight by $\delta$ which does not change the ordering, that is
$$ A_k(x_{i+1}) \le A_k(x_i)+ \delta  \le A_k(x_{i-1})$$
Each of these steps by $\delta$ causes the weight of $\{ x_j : j < i \}$ to increase (or decrease) by $\delta/2$ and the weight of $\{ x_j : j \le i \}$ to decrease (or increase) by $\delta/2$, for a total change of $\delta$.
Each record $x$ can be ``walked" to its new position through a sequence of such differences, equal in weight to some other record at each intermediate step, where the accumulation of $\delta$s is at most $|A_k(x) - A_k'(x)|$, for a total of $\|A - A'\|$.
\end{proof}

\section{\lowercase{w}PINQ operator implementation}\label{sec:opsImplement}

We give the flavor of a few operator implementations, pointing out differences from traditional incremental dataflow.

\myparab{SelectMany.  } The \texttt{SelectMany} operator is linear, and each change in input weight simply results in the corresponding change in output weight, for records determined by the supplied result selector.

\myparab{Union, Intersect, Concat, and Except. }
Both \texttt{Union} and \texttt{Intersect}  maintain dictionaries from records to weights, for each of their inputs. As a weight change arrives on either input, the corresponding two weights are retrieved and consulted to determine if the maximum or minimum weight has changed, and the appropriate differences are emitted and stored weights updated. \texttt{Concat} and \texttt{Except} pass through differences from either input, with weight negated in the case of \texttt{Except}'s second input.

\Snote{To Frank. In light of the modifications to Section~\ref{sec:imp}, you probably want to shorten/modify the following paragraph.  Also, the paragraph is fairly hard to parse, I think because I can't tell when you are describing the "normal" case of where $\|A_k\| + \|B_k\|$ does change, to the special "optimized" case where it does not.  It's also not clear to me if there is any new information here relative to what is in the end of Section~\ref{sec:imp}.  It think the only new thing here is the mention of the dictionary; if so, maybe we can cut this paragraph altogether, unless there is something particularly exciting about using a dictionary? or you can just mention the dictionary in Section~\ref{sec:imp}.}

\myparab{Join.  } For each input, \texttt{Join} operators maintain a dictionary from keys to a list of the records (and their weights) mapping to that key. When changes arrive on either input, the \texttt{Join} is tasked with emiting the differences in cross-product weights, as in \eqref{eq:joinweight}, which it can do explicitly by constructing both the old and new weighted output collection. In the case that the term $\|A_k\| + \|B_k\|$ has not changed (not uncommon for our random walks over graphs, where edges switch destinations rather than disappear) substantial optimization can be done. For input differences $a_k$ and $b_k$ the output difference equals
\begin{eqnarray*}
&&\texttt{Join}(A_k + a_k, B_k + b_k) -  \texttt{Join}(A_k, B_k) \\
& = & \frac{a_k \times B_k^T}{\|A_k\| + \|B_k\|} + \frac{A_k \times b_k^T}{\|A_k\| + \|B_k\|} + \frac{a_k \times b_k^T}{\|A_k\| + \|B_k\|} \; .
\end{eqnarray*}
We can avoid the potentially large term $A_k \times B_k^T$, whose denominator has remained the same. In traditional incremental dataflow \texttt{Join} does not scale its inputs, and this optimization would always be in effect.

\Snote{The use of the word ``retire'' below is confusing. Also, when you say ``\texttt{Shave} maintains a dictionary from records to their weight'', do you mean record+index? How to indices figure into how this operator is implemented?}

\myparab{Shave.  } \texttt{Shave} transforms each input record into a collection of unit weight output records paired with an index distinguishing each of the output records. To efficiently retire input differences, \texttt{Shave} maintains a dictionary from records to their weight, from which it can determine which output differences must be produced.

\Snote{This is not very clear.  Why does the ``ordering of records'' matter? Why do ``prefixes'' matter? }

\myparab{GroupBy.  } \texttt{GroupBy} groups records by key, and independently applies a reducer to each prefix of these records when ordered by weight. When presented with input difference, the ordering of records may change, and \texttt{GroupBy} operator must determine which prefixes must be re-evaluated. Consequently, \texttt{GroupBy} maintains a dictionary from keys to sorted lists of records (ordered descending by weight) mapping to that key. Input differences are retired by comparing the old and new sorted lists, and emitting the resulting records in difference.
\fi 

\vfill
\onecolumn
\section{Reproduction of Sala~\etal's result: Non-uniform noise for the JDD}\label{apx:jddProof}


\begin{claim} Let $D$ be the domain of possible node degrees, \ie $D=\{0,1,...,d_{\max}\}$. The following mechanism is $\epsilon$-differentially private: for each pair $(d_i,d_j) \in D \times D$, release a count of number of edges incident on nodes of degree $d_i$ and degree $d_j$, perturbed by zero-mean Laplace noise with parameter $4\max\{d_i,d_j\}/\epsilon$.
\end{claim}
\begin{proof}
We need to show that for two graphs $G_1$ and $G_2$, differing in a single edge $(a,b)$, that the mechanism releasing all pairs $(d_i,d_j)$ with Laplace noise of with parameter $4\max\{d_i,d_j\}/\epsilon$ satisfies $\epsilon$-differential privacy.

\myparab{Notation.}  This algorithm works on undirected edges, so that the pair $(d_i,d_j)=(d_j,d_i)$.
Let $n(i,j)=4\max\{d_i,d_j\}$.
Let $t_1(i,j)$ be the true value count of pair $(d_i,d_j)$ in graph $G_1$, and respectively $t_2(i,j)$ for $G_2$.

\smallskip
If we add Laplace noise $n(i,j)$ to each true count $t_1(i,j), t_2(i,j)$ the probability that we get a result where the $(d_i,d_j)$-th pair has values $r(i,j)$ is proportional to
\begin{align*}
\Pr[M(G_1) = r] &\propto \prod_{i,j} \exp\left(-  |r(i,j) - t_1(i,j)| \times \epsilon/ n(i,j)\right)
\end{align*}
Fixing the output $r$, we are interested in the ratio of this probability for $G_1$ and $G_2$. Fortunately, the constant of proportionality is the same for the two (because the Laplace distribution has the same integral no matter where it is centered), so the ratio of the two is just the ratio of the right hand side above:
\begin{align*}
\frac{\Pr[M(G_1) = r]}{\Pr[M(G2) = r]}
& =  \prod_{i,j} \exp((|r(i,j) - t_2(i,j)| - |r(i,j) - t_1(i,j)|) \times \epsilon/n(i,j))\\
& \leq \prod_{i,j} \exp(|t_2(i,j) - t_1(i,j)| \times \epsilon/n(i,j)) \qquad\text{(triangle inequality)}\\
& = \exp(\epsilon \times \sum_{i,j} \left|t_2(i,j) - t_1(i,j)\right| / n(i,j)
) 
\end{align*}
Thus, it suffices to show that
\begin{equation}\label{eq:jddshow}
\sum_{i,j} \left|t_2(i,j) - t_1(i,j)\right| / n(i,j) \leq 1
\end{equation}
Notice that we need only worry about cases where $t_2(i,j) \neq t_1(i,j)$.  Suppose that $G_1$ contains the edge $(a,b)$, while $G_2$ does not, and wlog assume that $d_a\geq d_b$, where $d_a$ and $d_b$ are the degrees of node $a$ and $b$ in graph $G_1$.  The differences between $t_1$ and $t_2$ are as follows:
\begin{itemize}
\item The count of pair $t_1(d_a,d_b)$ is one higher than $t_2(d_a,d_b)$.
\item Ignoring the $(a,b)$ edge, which we already accounted for above, node $a$ has $d_a-1$ other edges incident on it; it follows that there are at most $d_a-1$ pairs $(d_a,*)$ that are one greater in $G_1$ relative to $G_2$.  Similarly, there are at most $d_b-1$ pairs $(d_b,*)$ that are one greater in $G_1$ relative to $G_2$.
\item Furthermore, since the degree of node $a$ in $G_2$ is $d_a-1$, it follows that there are at most $d_a-1$ pairs $(d_a-1,*)$ that are one greater in $G_2$ relative to $G_1$. Similarly, there are at most $d_b-1$ pairs $(d_b-1,*)$ that are one greater in $G_2$ relative to $G_1$.
\end{itemize}
Notice that the total number of differences, in terms of the degree of $a,b$ in $G_1$ (before the edge is removed) is $2(d_a-1)+2(d_a-1)+1=2d_a+2d_b-3$. Note that if we express this in terms of the degree of $a,b$ in $G_2$, we replace $d_a$ with $d_a+1$ and similarly for $d_b$, so
we get a total of $2d_a+2d_b+1$ as in \cite{sala}.
It follows that the left side of (\ref{eq:jddshow}) becomes:
\begin{align*}
\sum_{(i,j)} \frac{\left|t_2(i,j) - t_1(i,j)\right|}{n(i,j)}
& \le \frac1{n(d_a,d_b)}+
\frac{d_a-1}{n(*,d_a)}+
\frac{d_a-1}{n(*,d_a-1)}+
\frac{d_b-1}{n(*,d_b)}+
\frac{d_b-1}{n(*,d_b-1)}\\
& \leq \frac1{4\max\{d_a,d_b\}}+
\frac{d_a-1}{4d_a}+
\frac{d_a-1}{4d_a-4}+
\frac{d_b-1}{4d_b}+
\frac{d_b-1}{4d_b-4}  \quad\text{(since }n(*,d_a)\geq 4d_a)\\
& \leq \frac1{4d_a}+
2\frac{d_a-1}{4d_a}+
2\frac{d_a-1}{4d_a-4}
  \quad\text{(since }d_a\ge d_b)\\
& = 1 -\tfrac{1}{4d_a}\\
&\leq 1
\end{align*}
which is what we set out to show in (\ref{eq:jddshow}).
%
\end{proof}


\end{document}